\documentclass[journal]{IEEEtran}
 \usepackage{cite}

\ifCLASSINFOpdf
   \usepackage[pdftex]{graphicx}
   \graphicspath{{img/pdf/}{img/jpeg/}}
   \DeclareGraphicsExtensions{.pdf,.jpeg,.png}
\else
   \usepackage[dvips]{graphicx}
   \graphicspath{{img/eps/}}
   \DeclareGraphicsExtensions{.eps}
\fi
\usepackage{todonotes}
\usepackage[cmex10]{amsmath}
\usepackage{booktabs}
\usepackage{amsfonts}
\usepackage{amssymb}
\usepackage{mathrsfs}
\usepackage{booktabs}

\usepackage[tight,footnotesize]{subfigure}
\usepackage{wasysym}
\usepackage{color}
\usepackage{epsfig} 
\usepackage{epstopdf}
\usepackage{float}
\usepackage{amsthm}
\usepackage{mathtools}
\usepackage{relsize}
\usepackage[english]{babel}
\usepackage[autostyle]{csquotes}
\definecolor{light-gray}{gray}{0.8}
\def\nb0{{\mathbf{0}}}
\def\nb1{{\mathbf{1}}}

\def\ncalL{{\mathcal{L}}}

\def\ncalN{{\mathcal{N}}}

\def\nbbE{{\mathbb{E}}}

\def\nbbP{{\mathbb{P}}}

\def\nbbR{{\mathbb{R}}}

\def\nrma{{\rm a}}
\def\nrmb{{\rm b}}
\def\nrmc{{\rm c}}
\def\nrmd{{\rm d}}

\newtheorem{lemma}{Lemma}

\newtheorem{ndef}{Definition}

\newtheorem{theorem}{Theorem}
\newtheorem{prop}{Proposition}
\newtheorem{cor}{Corollary}

\newtheorem{remark}{Remark}
\newtheorem{assumption}{Assumption}
	
\def\argmax{\operatorname{arg~max}}
\def\figref#1{Fig.\,\ref{#1}}%
\def\E{\mathbb{E}}

\def\P{\mathbb{P}}
\def\pc{\mathtt{P_c}}

\def\R{\mathbb{R}}

\def\T{\beta}							

\def\sir{\mathtt{SIR}}
\def\ase{\mathtt{ASE}}

\def\calL{\mathcal{L}}

\def\Bx{{\mathcal{B}}^x}
\def\Bxx{{\mathcal{B}}^{x_0}}
\def\sb{\sigma_\mathrm{b}}

\def\Nx{{\mathcal{N}}^x}

\def \mb{\bar{m}_\mathrm{b}}
\def\Ax{{\mathcal{A}}^x}
\def\Axx{{\mathcal{A}}^{x_0}}
\def\sa{\sigma_\mathrm{a}}

\def \ma{\bar{m}_\mathrm{a}}

\def \x{\nu_0}
\def \y{t_k}
\def \z{r}

\def \intra{I_{\mathrm{ intra}}}
\def \inter{I_{\mathrm{ inter}}}

\def\rx{z_{1}}
\def\ry{z_{2}}

\def\txi {a_0}

\def\htxi{h_{a_0}}

\def\a {a}
\def\b {b}
\def \aa {a}
\def \bb {b}

\def\ha {h_{a_{x}}}
\def\hb {h_{b_{x}}}
\def\haa {h_{a_{x_0}}}
\def\hbb {h_{b_{x_0}}}

\usepackage[font={small}]{caption}
\newcommand{\chb}[1]{{\color{black}#1}}
\newcommand{\chr}[1]{{\color{black}#1}}

\allowdisplaybreaks
\usepackage{balance}
\begin{document}


\title{Fundamentals of Cluster-Centric Content Placement in Cache-Enabled Device-to-Device Networks}

\author{Mehrnaz Afshang,~\IEEEmembership{Student Member,~IEEE},
{Harpreet S. Dhillon,~\IEEEmembership{Member,~IEEE}},
 and
{Peter Han Joo Chong,~\IEEEmembership{Member,~IEEE}}

\thanks{M.~Afshang is with Wireless@VT, Department of ECE, Virgina Tech, Blacksburg, VA, USA and is with School of EEE, Nanyang Technological University, Singapore.
Email: mehrnaz@vt.edu. H. S.~Dhillon  is with  Wireless@VT, Department of ECE, Virgina Tech, Blacksburg, VA, USA. Email: hdhillon@vt.edu. P. H. J.~Chong is with School of EEE, Nanyang Technological University, Singapore. Email: Ehjchong@ntu.edu.sg.}
\thanks{This work was presented in part at the IEEE Globecom workshops, San Diego, CA, 2015~\cite{AfsDhiC2015b}. \hfill Last updated: \today.}}

\maketitle

\begin{abstract}
This paper develops  a comprehensive analytical framework with foundations in stochastic geometry to characterize the performance of {\em cluster-centric} content placement in a cache-enabled device-to-device (D2D) network. Different from {\em device-centric} content placement, cluster-centric placement focuses on placing content in each cluster such that the collective performance of {\em all} the devices in each cluster is optimized. 
\chb{Modeling the locations of the devices by a Poisson cluster process, we define and analyze the performance for three general cases: (i) $k$-Tx case: receiver of interest is chosen uniformly at random in a cluster and its content of interest is available at the $k^{th}$ closest device to the cluster center, (ii) $\ell$-Rx case: receiver of interest is the $\ell^{th}$ closest device to the cluster center and its content of interest is available at a device chosen uniformly at random from the same cluster, and (iii) baseline case: the receiver of interest is chosen uniformly at random in a cluster and its content of interest is available at a device chosen independently and uniformly at random from the same cluster.}  
Easy-to-use expressions for the key performance metrics, such as coverage probability and area spectral efficiency ($\ase$) of the whole network, are derived for all three cases. Our analysis concretely demonstrates significant improvement in the network performance when the device on which content is cached or device requesting content from cache is biased to lie closer to the cluster center compared to baseline case. Based on this insight, we develop and analyze a new generative model for cluster-centric D2D networks that allows to study the effect of intra-cluster interfering devices that are more likely to lie closer to the cluster center.
\end{abstract}
\IEEEpeerreviewmaketitle

\section{Introduction}

 \IEEEPARstart {D}{riven} by the increasing mobile data traffic, cellular networks are undergoing unprecedented paradigm shift in the way data is delivered to the mobile users~\cite{Cisco2015}. A key component of this shift is  device-to-device (D2D) communication in which proximate devices can deliver content on demand to their nearby users, thus offloading traffic from often congested cellular networks~\cite{andrews2014will,golrezaei2013femtocaching,boccardi2014five,song2015wireless}. 
 This is facilitated by the spatiotemporal correlation in the content demanded i.e., repeated requests for the same content from different users across various time instants \cite{cha2007tube,fast2005creating, tadrous2014joint}. Storing the  popular files at the ``network edge'', such as in small cells, switching centers or handheld devices, termed {\em caching}, offers an excellent way to exploit this correlation in the content requested by the users~\cite{bastug2014living,maddah2014fundamental,ahlehagh2012video,shanmugam2013femtocaching,gitzenis2013asymptotic,BlaszczyszynG14}.
Cache-enabled D2D networks are attractive due to the possible linear increase of capacity with the number of devices that can locally cache data~\cite{golrezaei2013femtocaching,molisch2014caching,ji2013wireless,6787081,ji2014fundamental,Altieri,ji2015throughput}.

The performance of cache-enabled D2D networks fundamentally depends upon i) the locations of the devices, and ii) how is cache placed on these devices. For instance, consider the {\em device-centric} placement where the content is placed on a device close to the particular device that needs it. While this is certainly beneficial for the device with respect to which the content is placed,  high performance degradation can happen if
 another device in the network wants to access the same content from the device on which it was cached. \chb{As a result, we focus on the {\em cluster-centric} placement, where the
goal is to  improve  the collective performance of all the devices in the network measured in terms of the coverage probability and area spectral efficiency ($\ase$). This requires several new results for the coverage probability and $\ase$, where the receiving devices of interest and/or the devices that contain content of interest for these receivers are parametrized in terms of the cluster-center. These results are the main focus of this paper. }


\subsection{Related Work and Motivation}

Existing works on the modeling and analysis of D2D networks has taken two main directions. The first line of work focuses on characterizing the asymptotic scaling laws for cache-enabled D2D networks using the well-known {\em protocol model}; see \cite{6787081,ji2013wireless, ji2014fundamental,ji2015throughput} for a small subset. These works rely on a grid-based clustering model where the space is  tessellated into square cells with devices in each cell forming distinct clusters. While these works provide several key insights, the key limitation is the use of the protocol model, which assumes that the communication between two nodes is possible only if the intended receiver is: i) within collaboration distance of the intended transmitter, and ii) outside the interference range of all other simultaneously active transmitters \cite{gupta2000capacity}. The second line of work considers the so-called {\em physical model}, where the successful communication between two nodes is based on the received signal-to-interference-and-noise ratio, unlike the {\em protocol model} where it is simply based on the distance \cite{gupta2000capacity}. Tools from stochastic geometry have been used for the tractable characterization of the key physical layer metrics, such as the rate and coverage \cite{haenggi2012stochastic,baccelli2009stochastic,mukherjee2014analytical}. These tools have resulted in significant advancement in the tractable modeling and analysis of downlink and uplink cellular networks \cite{dhillon2012modeling,mukherjee2012distribution,novlan2013analytical,elsawy2014stochastic}. Motivated by this encouraging progress, there has recently been a surge of interest in applying these tools to the analysis of D2D networks. These works are discussed next. 

Depending upon the spectrum allocated for D2D transmissions, the D2D networks can be classified into two categories: {\em in-band} and {\em out-of-band}. In the in-band D2D networks, D2D and cellular networks coexist in the same spectrum. Using tools from stochastic geometry, several coexistence aspects, such as mode selection between cellular and D2D~\cite{lin2013comprehensive,elsawy2014analytical,zhu2015joint}, 
coexistence of D2D and unmanned aerial vehicle~\cite{mozaffari2015unmanned}, distributed caching in D2D networks \cite{ShDhiC2015a},
 and D2D interference management to protect cellular transmissions~\cite{feng2014tractable,sun2014d2d,george2014analytical,sakr2014cognitive,mungara2014spatial}, have been addressed. On the other hand, in the out-of-band D2D, as the name implies, orthogonal spectrum is allocated for D2D and cellular transmissions. For this setup, various aspects and applications of D2D networks, such as  multicast D2D transmissions~\cite{lin2013modeling}, D2D communication with network coding~\cite{pyattaev2015understanding}, and traffic offloading from cellular networks to D2D networks~\cite{andreev2014analyzing}, have been studied. 

To lend tractability, the common approach in all the above mentioned stochastic geometry-based works for D2D networks is to model the locations of D2D transmitters (D2D-Txs) as a Poisson Point Process (PPP), and the locations of D2D receivers (D2D-Rxs) via two approaches: i) D2D-Rxs lie at a fixed distance from their intended D2D-Txs \cite{lin2013comprehensive,feng2014tractable,sun2014d2d,george2014analytical,sakr2014cognitive,mungara2014spatial,mozaffari2015unmanned}, and ii) D2D-Rxs are uniformly distributed in a circular region around their intended D2D-Txs \cite{elsawy2014analytical,pyattaev2015understanding,andreev2014analyzing}. While these models provide several useful design insights, they suffer from a key shortcoming of not being able to capture the notion of {\em device clustering}, which is quite fundamental to the D2D network architecture \cite{altieri2014fundamental,6787081,zhang2014social,ji2013wireless ,ji2014fundamental,hu2014evaluating}. This shortcoming was addressed in our very recent work~\cite{MehrnazD2D1,AfsDhiC2015}, where we modeled the device locations by a Poisson cluster process to analyze the performance of {\em device-centric} content placement policies. \chb{In contrast, the current work takes a cluster-centric approach where the focus is on placing content so as to optimize the performance of the whole cluster rather than the individual devices.} 
The analysis involves the characterization of several new distance distributions in Poisson cluster processes. 
More details along with other main contributions are explained in detail below.%

\subsection{Contributions and Outcomes}

\subsubsection*{Tractable model for cache-enabled D2D networks}
We develop a realistic analytic framework to study the performance of cluster-centric content placement policies in a cache-enabled D2D network. Modeling the locations of the devices by a Poisson cluster process (in particular a slight variation of a Thomas cluster process) and using tools from stochastic geometry and stochastic orders, we first prove that the collective performance of all the devices in a given cluster in terms of coverage probability is improved when the content of interest for each device is placed at the device closest to the cluster center. This policy, however, may not always be feasible due to the limited storage capacity and/or the energy of the closest device. Besides, placing all the content on a single device limits frequency reuse within a cluster to one, which may not be optimal in terms of the network throughput. 
\chb{As a result, we explore more general scenarios in which the content is distributed across devices in the cluster. The analysis of such scenarios require new methodology where the receiving devices of interest and/or the devices that contain their content of interest are parameterized in terms of their location relative to the cluster center. This forms the main technical contribution of the paper. More details are provided next.}
\subsubsection*{$\ase$ and coverage probability analysis}  \chb{We derive easy to use expressions for coverage probability and $\ase$ for the following three general cases: (i) $k$-Tx case: receiver of interest is chosen uniformly at random in a cluster and its content of interest is available at the $k^{th}$ closest device to the cluster center, (ii) $\ell$-Rx case: receiver of interest is the $\ell^{th}$ closest device to the cluster center and its content of interest is available at a device chosen uniformly at random from the same cluster, and (iii) baseline case: the receiver of interest is chosen uniformly at random in a cluster and its content of interest is available at a device chosen independently and uniformly at random from the same cluster.} 
\chb{A common feature in these cases is the parameterization of the receiver of interest and/or its serving device in terms of its location relative to the cluster center. These results provide key insights into the performance of cluster-centric content placement policies. }
A key intermediate step in the analysis is the characterization of distances from the D2D-Rx of interest to its serving device, and intra- and inter-cluster interfering devices for these cases.
\subsubsection*{New generative model for cluster-centric D2D networks}
The analysis of $k$-Tx and $\ell$-Rx  cases described above shows that the network performance improves significantly when the device(s) on which the content is cached or the device(s) requesting content from cache are biased to lie closer to the cluster center. This means that besides the D2D link of interest, the intra-cluster interfering links may be more likely to have a transmitter or receiver closer to the cluster center. To study the effect of this behavior on the network performance, we propose a generative model in which the device locations follow a {\em double-variance Thomas cluster process}, where each cluster consists of a denser and a sparser subcluster. Sampling the locations of the transmitters or receivers uniformly at random from the denser subcluster allows us to model the above described {\em biasing} behavior fairly accurately.

\begin{table}
\centering{
\caption{Summary of notation}
\label{table:summarynotation}

\scalebox{.72}{%
\begin{tabular}{c|c}
  \hline
   \hline
  \textbf{Notation} & \textbf{Description}  \\
     \hline
  $\Phi_\nrmc; \lambda_\nrmc$ & An independent PPP modeling the locations of D2D cluster center,\\
                              & density of D2D cluster center \\
  \hline
   $x$  & The location of cluster center\\
 \hline
 $\a$, $\b$ & The relative location of cluster member form cluster center\\
 \hline
 $R$, $r$ & The serving distance, where a realization of $R$ is denoted by $r$\\
 \hline
 $\Nx$& Set of devices inside the cluster \\
    \hline
    $\Nx_{\rm t}, \Nx_{\rm r}$&  Set of possible transmitting and receiving devices\\
    \hline
     $N, N_{\tt t}, N_{\tt r}$& Number of  total, possible transmitting and   receiving devices  \\
    \hline
  $\Ax;\:  \bar{m}_\nrma$, $\Bx; \: \bar{m}_\nrmb$ &Set of simultaneously active devices inside the cluster with mean $\bar{m}_\nrma$ and $\bar{m}_\nrmb$\\
  \hline
 $\sigma_\nrma^2, \: \sigma_\nrmb^2$  & Scattering variance of  cluster member location around cluster center\\
 \hline
  $P_\nrmd$  & Transmit power of devices engage in D2D communications\\
 \hline
 $\alpha$  & Path loss exponent corresponding to the D2D link; $\alpha>2$\\
 \hline
 $\ha, \:\hb$  & Exponential fading  coefficients with mean unity \\
 \hline
 $\T$  & $\sir$ threshold for successful demodulation and decoding\\
 \hline
 $\pc$   & Coverage probability\\
 \hline
  $\mathtt{ASE}$  & Area spectral efficiency\\
  \hline
  \hline
\end{tabular}
}
}
\end{table}

\section{System Model} \label{sec:sysmod}
We consider a clustered D2D network where the content of interest for devices of a given cluster is cached in the same cluster. This is inspired by the fact that the popular content may vary significantly across clusters. For instance, users in a library may be interested in an entirely different set of files than the users in a sports bar. Besides, larger inter-cluster distances make it difficult to establish direct communication across clusters. Note that while our model is, in principle, extendible to include inter-cluster communication, we will limit our discussion to more relevant case where direct communication is only between two devices of the same cluster. More details on how the content is placed in the devices of a given cluster will be provided in Section~\ref{sec:contentplacement}.


\subsection{System Setup and Key Assumptions}
We model the locations of the devices by a Poisson cluster process in which the {\em parent points} are drawn from a PPP $\Phi_\nrmc $ with density $\lambda_\nrmc$ and the {\em offspring points} are independent and identically distributed (i.i.d.) around each parent point~\cite{DalVerB2003}. The parent points and offsprings will be henceforth referred to as \emph{cluster centers}  and \emph{cluster members} (or simply {\em devices}),  respectively. The cluster members (or devices) around each cluster center $x\in \Phi_\nrmc$ are sampled from an i.i.d. symmetric normal distribution with variance $\sigma_\nrma^2$ in $\R^2$.
Therefore, the density function of the location of a cluster member relative to the location of its cluster center, $\a\in R^2$, is
\begin{equation}
f_A(\a)=\frac{1}{2 \pi \sigma_\nrma^2 }\exp\left(-\frac{\|\a\|^2}{2 \sigma_\nrma^2}\right). 
\end{equation}
If the number of cluster members in each cluster is Poisson distributed, this setup corresponds to the well-known \emph{Thomas cluster process} \cite{ganti2009interference}. Note that we will put some restrictions on the number of cluster members to facilitate characterization of distance distributions in the sequel. Therefore, our setup can be interpreted as a variant of Thomas cluster process.


Denote the set of devices belonging to the cluster centered as $x \in \Phi_\nrmc$ by $\Nx$. Partition this set into two subsets of (i) possible transmitting devices denoted by $\Nx_{\rm t}$, and (ii) possible receiving devices denoted by $\Nx_{\rm r}$. Within each cluster, the set of simultaneously active transmitters is denoted by
$\Ax \subseteq \Nx_{\rm t}$ and hence the set of simultaneously active transmitters in the whole network can be expressed as:  $$\Psi=\cup_{x\in \Phi_\nrmc} \Ax. $$ To keep the model general,   
  we assume that the number of simultaneously active transmitters $|\Ax|$ is not necessarily the same for each cluster. More specifically, $|\Ax|$ is modeled as a Poisson distributed random variable with mean $\bar{m}_\nrma$.  
 
Without loss of generality, we focus on a randomly chosen cluster, termed \emph{representative cluster}, with its cluster center denoted by $x_0 \in \Phi_\nrmc$. 
For this cluster, we assume that the total number of devices is $|\ncalN^{x_0}|=N$, and the number of possible transmitting devices is $|\ncalN^{x_0}_{\rm t}|=N_{\rm t}$. This assumption is made to facilitate \emph{order statistics} arguments that appear in the characterization of distance distributions in the sequel. Note that for a meaningful analysis,   the link corresponding to the D2D-Rx of interest in the representative cluster needs to be active. Once the location of the D2D-Tx of interest is fixed, the set of other simultaneously active transmitters in the representative cluster is sampled uniformly at random from the remaining $N_{\rm t}-1$ positions. Therefore, it is assumed that the number of intra-cluster interfering devices is Poisson distributed with mean $\bar{m}_\nrma-1$ conditioned on the total being less than $N_{\rm t}-1$. As a result, the average number of active devices in the representative cluster is $\bar{m}_\nrma$, which is consistent with the assumption made above regarding the number of simultaneously active transmitters per cluster.

 \begin{figure}[t!]
\centering{
        \includegraphics[width=.7\linewidth]{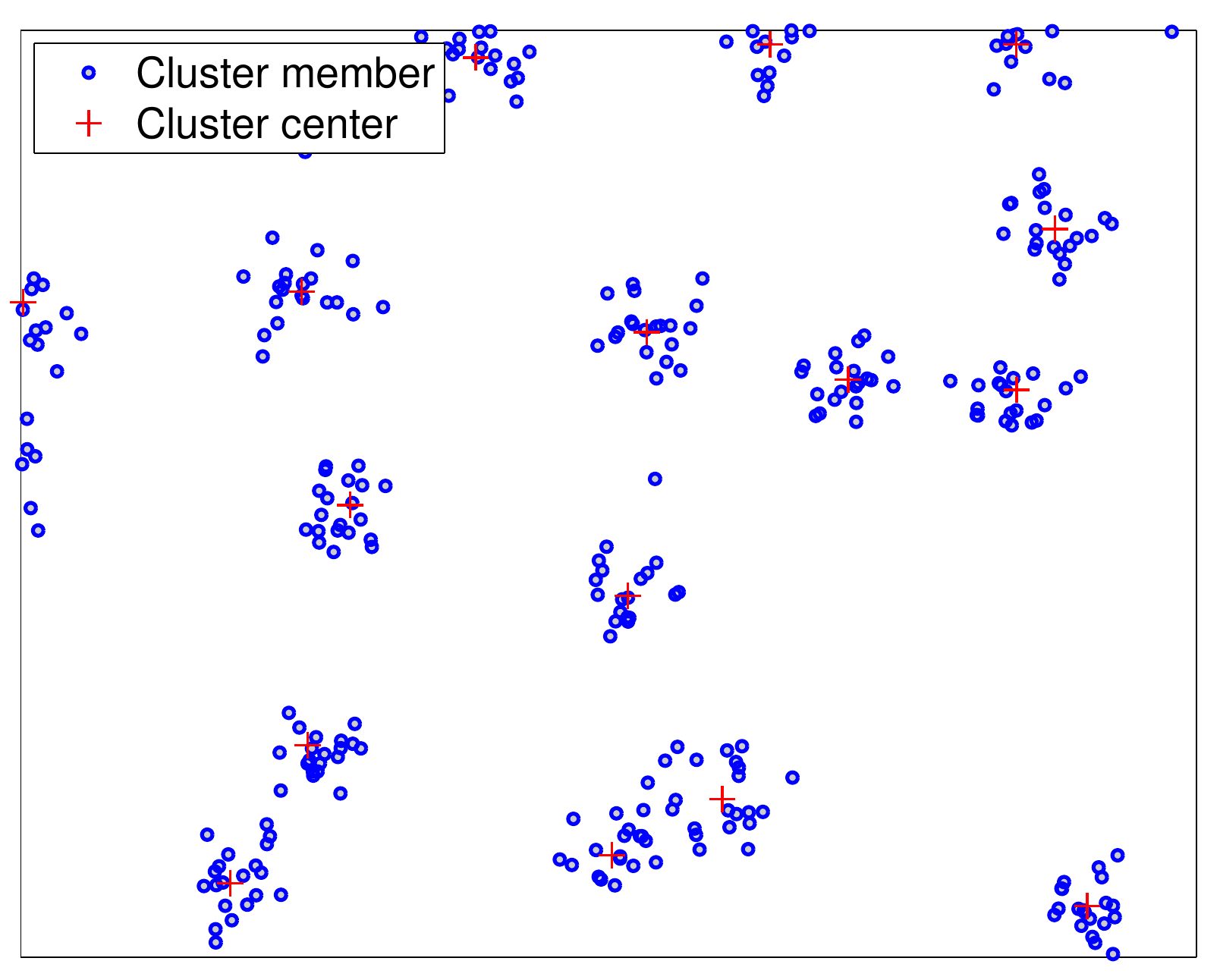}
              \caption{Illustration of D2D cluster network when cluster members (devices) are normally distributed around cluster center with $\sigma_\nrma=40$.}
                \label{Fig: Network Topo}
                }
\end{figure}
\subsection{Channel Model}
\label{subset: channel model}
Recall that the cluster center of the representative cluster is assumed to be located at $x_0 \in \Phi_c$, and hence the D2D-Rx of interest belongs to the set $\ncalN^{x_0}_{\rm r}$. Without loss of generality, the analysis is performed at the D2D-Rx of interest in the representative cluster, which is assumed to be located at the origin. 
The D2D-Txs are assumed to transmit at the constant power $P_\nrmd$. The content of interest for the D2D-Rx of interest is available at the device located at $\txi+x_0$, where $\txi$ indicates the location of this device relative to cluster center $x_0$.  Hence, the received power at the D2D-Rx of interest is
\begin{equation*}
P= P_\nrmd \htxi \|x_0+\txi\|^{-\alpha},
\end{equation*}
where $\htxi\sim \exp(1)$ models Rayleigh fading, and $\alpha$ is the power-law path loss exponent. 
The total interference experienced by the D2D-Rx of interest can be written as a sum of two independent terms. First, the interference from the set of devices inside the representative cluster, say \emph{ intra-cluster interference}, is given by
\begin{equation}
\label{eq: intra cluster interference}
\intra=\sum_{\aa \in \Axx \setminus \txi} P_\nrmd \haa\|x_0+ \aa\|^{-\alpha}.
\end{equation}
Second, the interference from the devices outside the representative cluster, say \emph{inter-cluster interference}, is given by
\begin{equation}
\label{eq: inter cluster interference}
 \inter=\sum_{x\in \Phi_{\nrmc}\setminus x_0}\sum_{\a \in \Ax} P_{\nrmd}\ha\|x+\a\|^{-\alpha}.
\end{equation}
Now, the signal-to-interference ratio $(\sir)$ at the D2D-Rx of interest at a distance $R=\|x_0+\txi\|$ from the serving device, where a realization of $R$ is denoted by $r$,  is:
\begin{equation}
\sir(r)=\frac{P_\nrmd \htxi r^{-\alpha}}{\inter+\intra}. 
\end{equation}
Note that the $\sir$ expression is not a function of the transmit power $P_\nrmd$ and therefore without loss of generality, we assume that $P_\nrmd=1$. 
\chb{For this setup, we study network performance in terms of coverage probability and $\ase$ which are formally defined next.
\begin{ndef}[Coverage probability]
The probability that $\sir$  of an arbitrary  link of interest at the receiver exceeds the required  threshold for successful demodulation and decoding.
\begin{align}
   \pc =  \E[\nb1\{\sir(r)>\T\}],
\end{align}
where $\T$ is a pre-determined threshold for successful demodulation and decoding at the receiver.
\label{Def coverage}
\end{ndef}
\begin{ndef}[Area spectral efficiency] The  average number of bits transmitted  per unit time
per unit bandwidth per unit area  can be defined as:
\begin{equation}
	\mathtt{ASE}= \lambda \log_2(1+\T) \mathbb{E}[\mathbf{1}\{\mathtt{SIR}(r) > \T\}],
	\label{eq:fixed_conditionalSE}
\end{equation}
where $\lambda$ is the number of simultaneously active transmitter per unit area.
\label{Def: ASE}
\end{ndef}
}


\section {Coverage Probability and $\ase$} \label{sec:contentplacement}
This is the first main technical section of the paper, where we first characterize the coverage-optimal cluster-centric content placement policy. We show that under this policy the content of interest for all the devices must be stored at the device closest to the cluster center. However, this may not be feasible due to storage and/or energy constraints of mobile devices. Besides, such a placement would limit the number of simultaneously active D2D connections over a given frequency band in a given cluster to one. Since aggressive frequency reuse is one of the advantages of D2D, this is clearly not desirable. \chb{As a result, we assume that the content is distributed across devices in a cluster. To enable the analysis of cluster-centric content placement policies for this setup, we define several cases of fundamental interest where the D2D-Rx of interest and/or the device which has its content of interest are parameterized in terms of their locations with respect to the cluster center. Easy-to-use expressions for coverage probability and $\ase$ are then derived for these cases.}
\subsection{Coverage-Optimal Content Placement}
\label{subsec: Coverage-Optimal Content Placement}
In this subsection, we study the coverage-optimal content placement problem in the proposed clustered D2D model. Note that while it is preferable to place the content required by each device at its immediately neighboring device, such {\em device-centric} content placement is not realistic. \chb{Therefore, we focus on the {\em cluster-centric} content placement, where the goal is to place the content in such a way that it improves the collective performance of the whole network. This can be achieved by fixing the {\em point of reference} for content placement to be the cluster-center instead of a particular receiver. To fix this key idea, we begin with a simple problem where we assume that the content of interest for the whole cluster is placed at a single device in $\ncalN^{x_0}_{\rm t}$ that is $k^{th}$ closest to the cluster center, where $k=1$ and $k=N_{\rm t}$ correspond to the closest and farthest devices from the set $\ncalN^{x_0}_{\rm t}$ to the cluster center, respectively. Our first goal is to find the value of $k$ that optimizes the performance of the whole cluster. We cast this problem as the coverage maximization problem, where coverage probability of a D2D-Rx of interest  is
\begin{align}
   \pc = \E[\nb1\{\sir(\|x_0+s_k\|)>\T\}],
   \label{eq:CC-Pc}
\end{align}
where $x_0$ is the location of the cluster center, and $s_k$ is the location of the $k^{th}$ closest device to the cluster center, which is also the serving device.
The optimal value of $k$ that maximizes this coverage probability is derived in the next Lemma.}
\begin{lemma}\label{lem: Optimal content placement2}
\chb{The optimal value of $k$ that maximizes the coverage probability given by \eqref{eq:CC-Pc} for the whole cluster is }
\begin{align}
\argmax_{k \in \{1,2,..., N_{\rm t}\}}\E[\nb1\{\sir(\|x_0+s_k\|)>\T\}]=1.%
\end{align}
\end{lemma}
\begin{proof}
See Appendix~\ref{proof: Optimal content placement}.

\end{proof}
\chb{
An intuitive interpretation of the above result is that all the devices in a given cluster should be served by a device that is {\em on an average} closest to all of them. As proved formally in the above Lemma, this device is the one that is closest to the cluster center. While this result is potentially useful in determining the coverage-optimal location of a cache-enabled small cell or a dedicated storage device, this simple policy limits the frequency reuse capability of D2D networks by concentrating all the content at a single device. Besides, such a policy may be infeasible due to storage and energy constraints of mobile devices. Therefore, it is important to distribute the content across multiple devices in a cluster. 

As noted above, for cluster-centric content placement, the point of reference will be the cluster-center instead of a particular receiver. This means the D2D-Rx of interest and/or the device that has cached its content of interest can be parametrized in terms of their locations with respect to the cluster center. For instance, it is precise to say that a receiver of interest will have its content of interest cached at a device that is $k^{th}$ closest to the cluster center from the set $\ncalN^{x_0}_{\rm t}$, where $k\in[1,N_{\rm t}]$. The value of $k$ will depend upon the content placement strategy being adopted, as discussed in the context of optimizing the total hit probability in Section~\ref{subsec: Total Hit Probability}. For performance comparison, we also consider a random placement strategy, where the content requested by the D2D-Rx of interest is available at a device chosen uniformly at random from the set $\ncalN^{x_0}_{\rm t}$. 

To analyze the performance of the above setup, we also need to define how the D2D-Rx of interest is chosen. We consider two choices: (i) D2D-Rx of interest is parameterized with respect to the cluster center as done for the transmitter above (say $\ell^{th}$ closest to the cluster center from the set $\ncalN^{x_0}_{\rm r}$), and (ii) the D2D-Rx of interest is chosen uniformly at random from the cluster. While the latter provides insights into the typical network performance, the former is useful in understanding how the performance of devices located towards the center of the cluster (small values of $\ell$) differs from those located at the edge of the cluster (large values of $\ell$). For this setup, we focus on the following three cases, each providing useful insights into the performance of D2D networks:
}

\begin{itemize}
\item {\em $k$-Tx case:}  In this case, the D2D-Rx of interest is chosen uniformly at random from a given cluster and the content of interest for this receiver is available at the $k^{th}$ closest transmitting device to the cluster center (in the set $\ncalN^{x_0}_{\rm t}$) from the same cluster. By tuning the value of $k$, we can study the effect of the location of the content/cache (relative to the cluster center) on the performance. 
%
\item {\em $\ell$-Rx case:} In this case, the D2D-Rx of interest is the $\ell^{th}$ closest device to the cluster center in the set $\ncalN^{x_0}_{\rm r}$ and its content of interest is available at a device chosen uniformly at random in $\ncalN^{x_0}_{\rm t}$. By tuning the value of $\ell$, we can understand how the performance of users located towards the center of the cluster differ from those located towards the cluster edge. 
\item {\em Baseline case:}  In the baseline case, we assume that the D2D-Rx of interest is chosen uniformly at random from $\ncalN^{x_0}_{\rm r}$, and the device containing its content of interest is also chosen uniformly at random from $\ncalN^{x_0}_{\rm t}$. This simple case will act as a baseline for performance comparisons.
\end{itemize}
\chb{Note that we can, in principle, define $k$-Tx $\ell$-Rx case, where the D2D-Rx of interest is the $\ell^{th}$ closest device to the cluster center from the set $\ncalN^{x_0}_{\rm r}$ and its content of interest is available at the $k^{th}$ closest device to the cluster center from the set  $\ncalN^{x_0}_{\rm t}$. Due to lack of space and the fact that the essence of this case will be captured approximately in the new {\em generative model} studied in Section~\ref{Sec:Future of D2D network}, we do not consider this explicitly. }

As discussed in detail in Section~\ref{sec:sysmod}, the effect of frequency reuse is studied by assuming that multiple D2D links in a given cluster can be activated simultaneously. In particular, once the location of the serving device is decided, the locations of intra-cluster interfering transmitters are sampled uniformly at random from the remaining points of $\ncalN^{x_0}_{\rm t}$. Note that the location of these interfering devices can be sampled in more sophisticated ways (e.g., biased to lie closer to the cluster center). This will be discussed in Section~\ref{Sec:Future of D2D network}. We now derive the coverage probabilities for the three cases described above in the next subsection. 

\subsection{Coverage Probability Analysis}
\label{sub:Coverage Probability Analysis }
\chb{Before going into the detailed analysis of coverage probability,  we characterize the distributions of the distances from intra- and inter-cluster devices  to the D2D-Rx of interest under various polices. Using  these distance distributions, we derive the  Laplace transform of distribution of intra- and inter-cluster interference distributions. As will be evident from our analysis, characterizing   Laplace transform of interference distribution is key intermediate result for the coverage probability analysis. }
The distances between the D2D receiver of interest and the various inter/intra cluster interfering devices are in general correlated. Focusing first on the intra-cluster devices, denote the distances from the D2D-Rx of interest to the intra-cluster devices by $\{w\}$, where $w=\|x_0+\aa\|$. Clearly these distances are correlated because of the common distance between the cluster center to the D2D-Rx of interest, $\nu_0=\|x_0\|$. As discussed in detail in \cite{MehrnazD2D1} for {\em device-centric} content placement, this correlation can be handled by conditioning on the common distance $\nu_0=\|x_0\|$, after which the distances $\{w\}$ become {\em conditionally} i.i.d., which lends tractability to the analysis of the Laplace transform of interference distribution, thus resulting in tractable expressions for coverage probability and $\ase$. For the current cluster model, the conditional  distance distribution  $f_W(w| \nu_0)$  is characterized by Rician distribution \cite{MehrnazD2D1}: 
\begin{equation}
\label{Eq: Rice distribution}
f_W(w| \nu_0)=\frac{w}{ \sigma_\nrma^2} \exp\left(-\frac{w^2+\nu_0^2}{2 \sigma_\nrma^2}\right) I_0\left(\frac{w \nu_0}{\sigma_\nrma^2}\right),\quad w>0.
\end{equation}
where $I_0(.)$ is  the modified Bessel function of the first kind with order zero. Since probability density function (PDF) of the Rician distribution will be frequently used in the sequel, we define its functional form below to simplify the notation.
\begin{ndef}(Rician distribution). The PDF of the Rician distribution $f_Y(y|z)$ is 
\begin{equation}
\mathtt{Ricepdf}(y,z;\sigma^2)=\frac{y}{ \sigma^2} \exp\left(-\frac{y^2+z^2}{2 \sigma^2}\right) I_0\left(\frac{y z}{\sigma^2}\right), \quad y>0,
\end{equation}
where   $\sigma$ is the scale parameter of the distribution.
\end{ndef}

Note that while the intra-cluster distances to the D2D-Rx of interest become conditionally i.i.d. by conditioning on $\nu_0=\|x_0\|$, there is still a possibility of dependence between the distances to the serving and interfering devices when the serving device is not chosen uniformly at random from the cluster. The easiest way to understand this dependence is by recalling that while the sequence of distances to intra-cluster devices, $w=\|x_0+\aa\|$, is conditionally i.i.d., the ``ordered'' choice of the serving device impacts the distribution of the remaining elements in the sequence (distances to the interfering devices). This is clearly true in the $k$-Tx case, where the serving device is chosen to be the $k^{th}$ closest device to the cluster center. However, it turns out that this dependence can be handled by first conditioning on the distance from the serving device to the D2D-Rx of interest, denoted by $t_k=\|s_k\|$, and then partitioning the intra-cluster interfering devices into two subsests: (i) devices that are closer than the serving device to the cluster center, denoted by $a\in\Axx_{\rm in}$, where the distance to the cluster center is denoted by $t_{\rm in}=\|a\|$, and (ii) devices that are farther than the serving device to the cluster center, denoted by $a\in \Axx_{\rm out}$, where the distance to the cluster center is denoted by $t_{\rm out}=\|a\|$. Please refer to  \figref{Fig: Sysmtem model case 1} for the pictorial representation. In the following Lemma, we prove that the distances from devices in $a\in\Axx_{\rm in}$ and $a\in \Axx_{\rm out}$ are respectively i.i.d., which lends tractability to the interference analysis. This result along with the conditional distribution of the distances is given next.
  \begin{figure}[t!]
\centering{
        \includegraphics[width=.7\linewidth]{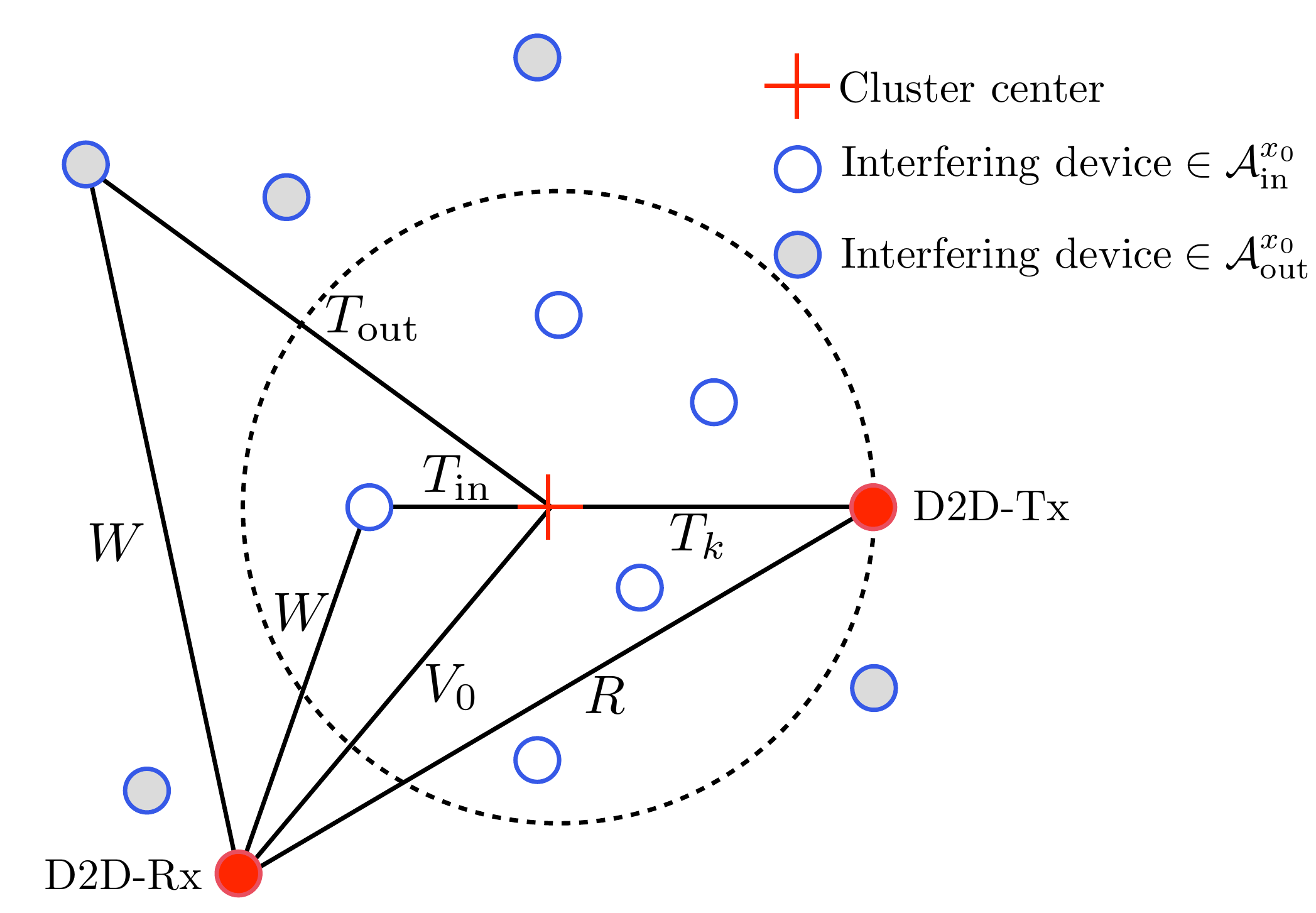}
              \caption{Illustration of intra-cluster devices for the $k$-Tx case.}
                \label{Fig: Sysmtem model case 1}
                }
\end{figure}

\begin{lemma}[Distance of intra-cluster interfering device to the D2D-Rx of interest in the $k$-Tx case]
The  distances from the intra-cluster interfering devices to the D2D-Rx of interest in the $k$-Tx case, i.e., $\{w=\|x_0+\aa\|\}$ are conditionally i.i.d., conditioned  on $\nu_0=\|x_0\|$, and $t=\|a\|$ (where $t$ can be either  $t_{\rm in}$ or $t_{\rm out}$), with PDF
\begin{equation}\label{eq: intra-cluster distance case 1}
f_W(w|\nu_0,t)= \frac{1}{\pi}\frac{w/\nu_0 t}{\sqrt{1-\Big(\frac{\nu_0^2+t^2-w^2}{2 \nu_0 t}\Big)^2}}, |\nu_0-t|<w<\nu_0+t,
\end{equation}
where, if $t=t_{\rm in}$, 
\begin{equation} 
f_{T_{\rm in}}(t_{\rm in}|t_k)=\left\{
 \begin{array}{cc}
 \frac{\frac{t_{\rm in}}{  \sigma_\nrma ^2}\exp\left(-\frac{t_{\rm in}^2}{2 \sigma_\nrma^2}\right)}{1-\exp\left(-\frac{t_k^2}{2 \sigma_\nrma^2}\right)}, & t_{\rm in}<t_k\\
 0, & t_{\rm in}\geq t_k
 \end{array}\right.,
 \end{equation}
if $t=t_{\rm out}$, then,
\begin{equation} 
f_{T_{\rm out}}(t_{\rm out}|t_k)=\left\{
 \begin{array}{cc}
 \frac{\frac{t_{\rm out}}{  \sigma_\nrma ^2}\exp\left(-\frac{t_{\rm out}^2}{2 \sigma_\nrma^2}\right)}{\exp\left(-\frac{t_k^2}{2 \sigma_\nrma^2}\right)}, & t_{\rm out}>t_k\\
 0, & t_{\rm out}\leq t_k
 \end{array}\right.
 \end{equation}
\label{lem: distance from intra-cluster device to typical device in case 1}
\end{lemma}
\begin{proof}
See Appendix \ref{App: proof of distance from intra-cluster device to typical device in case 1}.
\end{proof}
Using this distance distribution, the exact expression of the conditional Laplace transform of intra-cluster interference distribution in  the $k$-Tx case is given next. \chr{Please note that the corresponding result appearing in the shorter version of this paper~\cite{AfsDhiC2015b} is an approximation.}
\begin{lemma} In the $k$-Tx case, the conditional Laplace transform of distribution of intra-cluster interference  \eqref{eq: intra cluster interference}, conditioned on $\nu_0$, where content of interest  is placed at distance $t_k=\|s_k\|$ from cluster center  is $\ncalL_{\intra}(s,t_k|\nu_0)=$
\begin{align}
\notag
\sum_{n=0}^{N_{\rm t}-1}\sum_{l=0}^{g_{\rm m}} \frac{{n \choose l}p^l(1-p)^{n-l}}{I({1-p};n-g_{\rm m},1+g_{\rm m})} M_{\rm in}(s,t_k|\nu_0)^l\\
\times M_{\rm out}(s,t_k|\nu_0)^{n-l}\frac{(\bar{m}_\nrma-1)^n e^{-(\bar{m}_\nrma-1)}}{n! \xi}
\label{eq: Laplace intra case 1}
\end{align}
with,
\small
\begin{align*}
&M_{\rm in}(s,t_k|\nu_0)= \int_{0}^{t_k} \int_{w_{\rm in}^{\rm L}}^{w_{\rm in}^{\rm U}} \frac{f_{W}(w|\nu_0,t_{\rm in})}{1+s w^{-\alpha}}f_{T_{\rm in}}(t_{\rm in}|t_k) \nrmd w \nrmd t_{\rm in},\\
&M_{\rm out}(s,t_k|\nu_0)=\int_{t_k}^{\infty} \int_{w_{\rm out}^{\rm L}}^{w_{\rm out}^{\rm U}} \frac{f_{W}(w|\nu_0,t_{\rm out})}{1+s w^{-\alpha}}f_{T_{\rm out}}(t_{\rm out}|t_k) \nrmd w \nrmd t_{\rm out},
\end{align*}
\normalsize
where $w_{\rm in}^{\rm L}=|\nu-t_{\rm in}|$, $w_{\rm in}^{\rm U}=\nu_0+t_{\rm in}$, $w_{\rm out}^{\rm L}=|\nu_0-t_{\rm out}|$, $w_{\rm out}^{\rm U}=\nu_0+t_{\rm out}$, $p=\frac{k-1}{N_{\rm t}-1}$, $g_{\rm m}=\min(n,k-1)$, $\xi=\sum_{j=0}^{N_{\rm t}-1}\frac{(\bar{m}_\nrma-1)^j e^{-(\bar{m}_\nrma-1)}}{j!}$, ${I({1-p};n-g_{\rm m},1+g_{\rm m})} $ is  regularized incomplete beta function, and density functions  of $f_W(w|\nu_0,t)$, $f_{T_{\rm in}}(t_{\rm in}|t_k)$, and $f_{T_{\rm out}}(t_{\rm out}|t_k)$ are given by Lemma \ref{lem: distance from intra-cluster device to typical device in case 1}.
\label{lem: Laplace exact case 1}
\end{lemma}
\begin{proof}
See Appendix \ref{proof: lemma Laplace intra case 1}.
\end{proof}
\chb{For the other two cases ($\ell$-Rx and baseline), the selection of serving device is done uniformly at random, which does not induce any dependence in the distances from the serving and interfering devices, leading to a simpler expression for the Laplace transform of intra-cluster interference distribution for these cases. Note that conditioning on $ \nu_0=\|x_0\|$ is still necessary to handle the correlation induced by the common term $x_0$ in the intra-cluster distances, as discussed earlier in this section.}
\begin{lemma}  For the $\ell$-Rx and baseline cases, the conditional Laplace transform of the intra-cluster interference distribution,  conditioned on the distance from the D2D-Rx of interest to the  cluster center, $\nu_0=\|x_0\|$, is
\begin{align}\label{Eq: laplace intra typical _sum}
 \calL_{\intra} (s |\nu_0)=\sum_{n=0}^{N_{\rm t}-1} \Big ( M(w|\nu_0)\Big)^n  
 \frac{(\bar{m}_\nrma-1)^n e^{-(\bar{m}_\nrma-1)}}{n! \xi}
\end{align}
with   $M(w|\nu_0)=\int_0^\infty \frac{1}{1+s w^{-\alpha}}f_{W}(w|\nu_0)\nrmd w$. Assuming $\bar{m}_\nrma\ll N_{\rm t}$, we have
 \begin{align}\label{Eq: laplace intra typical}
 \calL_{\intra} (s |\nu_0)\simeq\frac{1}{\xi} \exp\Big(-(\bar{m}_\nrma-1)(1-M(w|\nu_0))\Big),
\end{align}
 where $\xi=\sum_{j=0}^{N_{\rm t}-1}\frac{(\bar{m}_\nrma-1)^j e^{-(\bar{m}_\nrma-1)}}{j!}$, and $f_{W}(w|\nu_0)= \mathtt{Ricepdf}(w,\nu_0;\sigma_\nrma^2)$.
\label{lem: lap intra typical}
\end{lemma}
\begin{proof}
See Appendix \ref{proof: lemma Laplace intra}.
\end{proof}
\begin{remark}
As discussed in the sequel, the optimal number of simultaneously active links is much smaller than the total number of potential transmitters. Hence, the assumption $\bar{m}_\nrma \ll N_{\rm t}$ taken to derive the simpler expression~\eqref{Eq: laplace intra typical} in Lemma~\ref{lem: lap intra typical} is fairly reasonable. As a result, \eqref{Eq: laplace intra typical} will in general be quite accurate.
\end{remark}
\begin{remark} \label{Rem:kTxApprox}
While Lemma~\ref{lem: lap intra typical} is exact for $\ell$-Rx and baseline cases, it also provides a tractable approximation for the $k$-Tx case, whose exact expression for the Laplace transform of intra-cluster interference distribution given by Lemma~\ref{lem: Laplace exact case 1} is much more complicated due to the presence of two summations. Lemma~\ref{lem: lap intra typical} is an approximation for the $k$-Tx case because it ignores the effect of ``ordered'' selection of serving device on the distance distributions of the intra-cluster interference devices. However, this approximation for the $k$-Tx case will be numerically shown to be quite tight in Section~\ref{sec:NumResults}. 
\end{remark}

The Laplace transform of intra-cluster interference distribution given by Lemma~\ref{lem: lap intra typical} can be simplified further under the following assumption without loosing much accuracy.
\begin{assumption}[Un-correlated intra-cluster distances assumption for $k$-Tx and baseline cases]\label{Ass: Identical intra-cluster distances} 
Since the devices are normally distributed around the cluster centers, the distances from the intra-cluster devices to the D2D-Rx of interest are Rayleigh distributed with the following PDF when the D2D-Rx of interest is chosen uniformly at random~\cite{MehrnazD2D1} 
\begin{equation} 
  f_{W}(w)= \frac{w}{2 \sigma_{\nrma} ^2}\exp\left(-\frac{w^2}{4 \sigma_{\nrma}^2}\right), \quad w>0.
  \label{Eq: Intra_cluster distance _typical approximation}
    \end{equation}
However, as discussed earlier in this section, the distances are correlated due to the presence of the common distance $\nu_o=\|x_0\|$, due to which we conditioned on this distance in \eqref{Eq: Rice distribution}. However, if we ignore this correlation, we can simplify the analysis by assuming that the distances are i.i.d. Rayleigh distributed with the PDF given by~\eqref{Eq: Intra_cluster distance _typical approximation}. This approximation is however not applicable for the $\ell$-Rx case where the D2D-Rx of interest is not chosen uniformly at random.
\end{assumption}
Under this assumption the approximation for Laplace transform
of intra-cluster interference distribution is given next. It is applicable for the $k$-Tx and baseline cases.
\begin{cor}
  Under Assumption \ref{Ass: Identical intra-cluster distances}, the Laplace transform of intra-cluster interference distribution in $k$-Tx and baseline cases is
\begin{multline}\label{Eq: app laplace intra typical}
 \tilde{\calL}_{\intra} (s)=\frac{1}{\xi}      \exp\Big(-(\bar{m}_\nrma-1)\int_0^\infty \frac{s w^{-\alpha}}{1+s w^{-\alpha}}\\
  f_{W}(w)\nrmd w\Big) , \quad
\end{multline}
where    $\xi=\sum_{j=0}^{N_{\tt t}}\frac{(\bar{m}_\nrma-1)^j e^{-(\bar{m}_\nrma-1)}}{j!}$ and $f_{W}(w)$ given by \eqref{Eq: Intra_cluster distance _typical approximation}.
\label{cor: app lap intra}
\end{cor}
We will use this simpler expression to provide easy to compute expression for coverage probability later in this section. 
We now derive the Laplace transform of inter-cluster interference distribution. Recall that the inter-cluster interferers are sampled uniformly at random in all three cases, which means the following result is exact for all three cases.
\begin{lemma}\label{Lem: Lap_Inter} For all three cases, the Laplace transform of distribution of inter-cluster interference at D2D-Rx of interest in~\eqref{eq: inter cluster interference} is
\begin{multline}\label{Eq: Lap_Inter}
\calL_{\inter} (s) =\exp\Big(-2 \pi \lambda_\nrmc\int_0^\infty \Big(1-\exp\Big(-\bar{m}_\nrma\int_0^\infty \frac{s u^{-\alpha}}{1+s u^{-\alpha}}
 \\f_U(u|\nu)\nrmd u \Big)\nu \nrmd \nu\Big)\Big),
\end{multline}
\normalsize
 where  $f_{U}(u| \nu)=\mathtt{Ricepdf}(u,\nu;\sigma_{\nrma}^2)$.
\end{lemma}
\begin{proof}
The Laplace transform of the inter-cluster interference distribution is special case of the Lemma \ref{lem: Future Laplace inter}. Since the proof follows on the same lines as that of Lemma \ref{lem: Future Laplace inter}, its skipped. \end{proof}

\subsubsection{Coverage probability analysis of $k$-Tx case}
Recall that the D2D-Rx of interest in this case is chosen uniformly at random and the D2D-Tx of interest is the $k^{th}$ closest transmitting device to the cluster center (in the set $\ncalN^{x_0}_{\rm t}$). We first derive the serving distance distribution for this case. 
\begin{lemma}
The PDF of the serving distance, i.e., $r=\|x_0+s_k\|$, conditioned on the distances $\nu_0=\|x_0\|$ and $t_k=\|s_k\|$ for the $k$-Tx case is 
\begin{align}
f_{R}(r|\x,\y)&= \frac{1}{\pi} \frac{\z/{\x,\y}}{\sqrt{1-\left(\frac{\x^2+\y^2-\z^2}{2\x\y}\right)^2}}, \:\: |\x-\y|<\z<\x+\y, \ \ \ \label{eq: fR case 1}
\end{align}
with
\small
\begin{align}
f_{V_0}(\x)&=\frac{\x}{  \sigma ^2}\exp\left(-\frac{\x^2}{2 \sigma_\nrma^2}\right), \quad \x>0 \label{eq: fV case 1}\\
f_{T_k}(\y)&=\frac{N_{\rm t}!}{(k-1)!(N_{\rm t}-k)!}{F(\y)}^{k-1}(1-F(\y))^{N_{\rm t}-k} f(\y) \label{eq: fT case1}
\end{align}
\normalsize
 where  $f(\y)=\frac{\y}{  \sigma ^2}\exp(-\frac{\y^2}{2 \sigma_\nrma^2})$, and $F(\y)=1-\exp(-\frac{\y^2}{2 \sigma_\nrma^2})$.
 \label{lem: serving dist case1}
\end{lemma}
\begin{proof}
The PDF of serving distance $r=\|x_0+s_k\|$ conditioned on the $v_0$ and $t_k$, i.e., $f_R(r|\x,\y)$ can be derived exactly on the same lines as $f_{W}(w|\nu_0, t)$ given by \eqref{eq: intra-cluster distance case 1}. Hence, the proof  is  skipped. Here, $f_{V_0}(\x)$ is Rayleigh distributed owing to the fact that the D2D-Rx of interest is a randomly chosen device where  devices are normally scattered around the cluster center. Finally, for $f_{T_k}$ note that the distances of intra-cluster devices to the cluster center are i.i.d. Rayleigh distributed with $T_k$ being the $k^{th}$ smallest sample out of $N_{\rm t}$ elements, whose distribution follows by order statistics (see~\cite[eq (3)]{david1970order}).
\end{proof}
Using this result, we now derive the coverage probability for the $k$-Tx case in the following theorem.
\begin{theorem}[Coverage probability: $k$-Tx case] Using Laplace transform of distribution of interference in \eqref{eq: Laplace intra case 1}, and \eqref{Eq: Lap_Inter}, the coverage probability of the D2D-Rx of interest is  
\begin{align}\notag
{{\tt P}_{c_k}^{\rm Tx}}=\int_0^{\infty}\int_0^{\infty}  \int_{r^{\rm L}}^{r^{\rm U}} &\calL_{\inter} (\T r^{\alpha} ) \calL_{\intra} (\T r^\alpha ,t_k|\nu_0)f_{R}(r|\x,\y)\\
 &\times f_{V_0}(\x)f_{T_k}(\y) \nrmd r \nrmd\nu_0 \nrmd \y, \label{Eq: Pc case 1}
\end{align}
with $r^{\rm L}= |v-t_k|$, and $r^{\rm U}= v+t_k$, where $f_{R}(r|\x,\y)$, $f_{V_0}(\x)$, and $f_{T_k}(\y)$ are given by \eqref{eq: fR case 1}, \eqref{eq: fV case 1}, and \eqref{eq: fT case1} respectively.
\label{Thm: Coverage case1 }
\end{theorem}
\begin{proof}
From the definition of coverage probability, we have
\begin{align}
{{\tt P}_{c_k}^{\rm Tx}} &=\mathbb{E}_{T_k} \mathbb{E}_{V_0} \mathbb{E}_R\left[ \nbbP \left\{ h_{0x_0} > \T r^\alpha (\inter+\intra) \,\Big|\,R,V_0,T_k  \right\} \right] \nonumber \\
&\stackrel{(a)}{=} \mathbb{E}_{T_k} \mathbb{E}_{V_0} \nbbE_R\left[ \nbbE\left[\exp\left(- \T r^\alpha (\inter+\intra) \right) \Big|\,R,V_0,T_k \right]\right]\nonumber
\end{align}
where $(a)$ follows from $h_{0x_0}\sim \exp(1)$.  {The result follows from the fact  that intra- and inter-cluster interference powers  are independent,} followed by  the  expectation over  $R$ given $\nu_0$ and $t_k$, followed by expectation over $V_0$ and $T_k$. The PDFs of $V_0$ and $T_k$ are given by \eqref{eq: fV case 1} and \eqref{eq: fT case1},  respectively.
\end{proof}
As discussed in Remark~\ref{Rem:kTxApprox}, the exact expression for the Laplace transform of the intra-cluster interference distribution given by \eqref{eq: Laplace intra case 1} in Lemma~\ref{lem: Laplace exact case 1} is quite complicated due to the presence of two summations. To improve tractability, the simpler expression of Lemma~\ref{lem: lap intra typical} can be used. This leads to an approximation since the dependence of the distances from the intra-cluster interfering devices on the selection of the serving device is not captured. The approximate result is given next. The proof follows on the same line as that of Theorem~\ref{Thm: Coverage case1 }.

%
\begin{cor} Using Laplace transform of intra-cluster interference distribution given by Lemma \ref{lem: lap intra typical}, the coverage probability of $k$-Tx case can be approximated as
\begin{multline}\label{eq: coverage case 1}
{{\tt P}_{c_k}^{\rm Tx}} \simeq \int_0^{\infty}\int_0^{\infty}  \int_0^{\infty} \calL_{\inter} (\T r^{\alpha} ) \calL_{\intra} (\T r^\alpha |\nu_0)f_{R}(r|\x,\y)\\
 \times f_{V_0}(\x)f_{T_k}(\y) \nrmd r \nrmd\nu_0 \nrmd \y,
\end{multline}
where $\calL_{\inter}(.)$ is given by \eqref{Eq: Lap_Inter},   and $f_{R}(r|\x,\y)$, $f_{V_0}(\x)$, $f_{T_k}(\y)$ are given by \eqref{eq: fR case 1}, \eqref{eq: fV case 1}, and \eqref{eq: fT case1} respectively.
\label{corr: approx Coverage case1 }
\end{cor}
 Although   the above coverage probability expression for $k$-Tx case seems to be  involved,  it can be easily evaluated by Quasi-Monte Carlo numerical integration methods (because the integrations are essentially expectations) \cite{caflisch1998monte}. 
  Using the approximation of the Laplace transform of the intra-cluster interference distribution  given by Corollary \ref{cor: app lap intra}, we  can further simplify  coverage probability expression in the next corollary.
\begin{cor}
By ignoring  intra-cluster distance correlations under Assumption \ref{Ass: Identical intra-cluster distances}, the coverage probability of the $k$-Tx case can be approximated as 
 \begin{equation}
 {{\tt P}_{c_k}^{\rm Tx}} \simeq \int_0^{\infty}\int_0^{\infty} \calL_{\inter} (\T r^{\alpha} ) \tilde{\calL}_{\intra} (\T r^\alpha )f_R(r| \y) f_{T_k}(\y) \nrmd r \nrmd \y
 \end{equation}
 \normalsize
 where $f_R(r| \y)= \mathtt{Ricepdf}(r,\y;\sigma_\nrma^2)$, and
 $f_{T_k}(t_k)$ is given by \eqref{eq: fT case1}.
\label{corr: app coverage case 1}
\end{cor}
\begin{proof}
See Appendix \ref{proof:  Corollary app coverage case1}.
\end{proof}
The tightness of the approximation will be validated in the numerical results section (Section~\ref{sec:NumResults}).

\subsubsection{Coverage probability analysis of $\ell$-Rx case}
\label{subsec: Coverage probability analysis of  Case 2 }
We now derive the coverage probability for the $\ell$-Rx case, where the D2D-Rx of interest is the $\ell^{th}$ closest device to the cluster center from the set $\ncalN^{x_0}_{\rm r}$ and its serving device is chosen uniformly at random from the set $\ncalN^{x_0}_{\rm t}$. 
\begin{theorem}[Coverage probability: $\ell$-Rx case] The coverage probability in $\ell$-Rx case is 
\small
\begin{align} \label{eq: coverage case 2}
{{\tt P}_{c_\ell}^{\rm Rx}}=  \int_0^\infty  \int_0^\infty \ncalL_{\inter}(\T r^{\alpha})\ncalL_{\intra}(\T r^{\alpha}|t_\ell)  f_R(r| t_\ell)  f_{T_\ell}(t_\ell) \nrmd r \nrmd t_\ell,
\end{align}
\normalsize with $f_{T_\ell}(t_\ell)=\frac{N_{\rm r}!}{(\ell-1)!(N_{\rm r}-\ell)!}{F(t_\ell)}^{\ell-1}(1-F(t_\ell))^{N_{\rm r}-\ell} f(t_\ell)$, where $f(t_\ell)=\frac{t_\ell}{  \sigma_\nrma ^2}\exp(-\frac{t_\ell^2}{2 \sigma_\nrma^2})$, $F(t_\ell)=1-\exp(-\frac{t_\ell^2}{2 \sigma_\nrma^2})$,  and $f_R(r| t_\ell)= \mathtt{Ricepdf}(r,t_\ell;\sigma_\nrma^2)$.
\label{Thm: Coverage case2 }
\end{theorem}
\begin{proof}
The proof follows on the same lines as the proof of Theorem \ref{Thm: Coverage case1 }. By definition, the coverage probability is
\begin{equation}
{{\tt P}_{c_\ell}^{\rm Rx}}=\nbbE_{T_\ell} \nbbE_R \left[\nbbE\left[\exp\left(- \T r^\alpha (\inter+\intra) \right) \Big| R, T_\ell \right]\right],
\end{equation}
where the PDF of serving distance $r=\|s_\ell+\aa\|$ conditioned on $t_\ell=\|s_\ell\|$ can be derived on the same lines as the PDF of the serving distance in Corollary \ref{corr: app coverage case 1}. This is because in this case, file of interest is available inside the cluster uniformly at random and D2D-Rx of interest is $\ell^{th}$ closest receiving device to the cluster center.  Thus, the D2D-Tx of interest is located at  $x_0+\aa$ where $x_0=s_\ell$, and $\aa$ is sampled from zero-mean complex Gaussian random variable.  
\end{proof}
\subsubsection{Coverage probability analysis of the baseline case}
\label{subsec: Coverage probability analysis of  Case 3 }
In the baseline case, we assume that both the D2D-Rx of interest and D2D-Tx of interest are sampled uniformly at random. The coverage probability for this case is given next. \chb{For the complete proof, please refer to \cite{MehrnazD2D1}, where the same case was used as the baseline case for device-centric content placement strategies}. Here we just provide a proof sketch.
\begin{theorem} [Coverage probability: baseline case]
\label{thm: coverage case 3}
The coverage probability in the baseline case is  

\small
\begin{align}\label{eq: coverage case 3}
 {{\tt P}_{c}^{\rm B}}=  \int_0^\infty  \int_0^\infty &\ncalL_{\inter}(\T r^{\alpha})\ncalL_{\intra}(\T r^{\alpha}|\nu_0) 
  f_R(r| \nu_0) f_{V_0}(\nu_0) \nrmd r \nrmd \nu_0,
\end{align}
\normalsize
where
 $ f_R(r| \nu_0)=\mathtt{Ricepdf}(r,\nu_0;\sigma_\nrma^2),$
and $f_{V_0}(\x)= \frac{\x}{  \sigma_\nrma ^2}\exp\left(-\frac{\x^2}{2 \sigma_\nrma^2}\right).$
\label{Thm: Coverage case3 }
\end{theorem}
\begin{proof}
The proof follows on the same lines as Theorem \ref{Thm: Coverage case1 }. By definition, the coverage probability is
\begin{equation}
 {{\tt P}_{c}^{\rm B}}=\nbbE_{V_0} \nbbE_R\left[ \nbbE\left[\exp\left(- \T r^\alpha (\inter+\intra) \right) \Big| R, V_0 \right]\right],
\end{equation}
where the PDF of  serving distance $r=\|x_0+a_0\|$  conditioned on $\nu_0=\|x_0\|$  is Rician distributed \cite{MehrnazD2D1}. Further, the D2D-Rx of interest is chosen uniformly at random,  which is sampled from a Gaussian distribution in $\nbbR^2$, and hence $\nu_0=\|x_0\|$ simply follows  Rayleigh distribution. 
\end{proof}

\subsection{Area  Spectral Efficiency Analysis}
After studying the coverage probability for all three policies in the previous subsection, we now focus on the area spectral efficiency ($\ase$), which is  defined in Definition \ref{Def: ASE}.
%
%
This definition is specialized to our  setup in the following proposition.

\begin{prop}
The $\ase$ for the three cases is given by:
\begin{equation}\label{eq: ASE}
    \mathtt{ASE}=\bar{m}_\nrma\lambda_\nrmc \log_2(1+\T) \pc,
\end{equation}
where $\pc$ is given by~\eqref{Eq: Pc case 1}, \eqref{eq: coverage case 2}, and~\eqref{eq: coverage case 3} for $k$-Tx, $\ell$-Rx, and baseline cases respectively. Here, $\bar{m}_\nrma \lambda_\nrmc$ denotes the average number of simultaneously active transmitting devices.
\end{prop}

\begin{remark} [Trade-off between channel orthogonalization and more aggressive spectrum reuse]
Intra-cell channel orthogonalization, i.e., a small number of simultaneously active devices per cluster, reduces intra-cluster interference at the expense of less spectrum reuse. We cast this classical tradeoff between higher number of simultaneously  active links (i.e., more spectrum reuse) and higher interference as the problem of finding the optimum $\bar{m}_\nrma$ that maximizes $\ase$:
\begin{equation}\label{Eq: ASE optimization}
  \mathtt{ASE}^*= \max _{\bar{m}_\nrma \in {1,...,N_{\tt t}}} \bar{m}_\nrma \lambda_\nrmc \log_2(1+\T) \pc.
\end{equation}
We will revisit this trade-off over the number of simultaneously active links in the numerical results section.
\chb{ By  solving this $\ase$ optimization problem numerically, we will demonstrate that optimum number of simultaneously active links is significantly different for the three cases, which means it is highly dependent on the choice of content placement policy.}
  \end{remark}

\section {New generative model for the cluster-centric D2D networks}
\label{Sec:Future of D2D network}
A key takeaway from the analyses of $k$-Tx and $\ell$-Rx cases, which will become more apparent in the numerical results discussed in Section~\ref{sec:NumResults}, is that the network performance improves significantly when the device(s) on which the content is cached or the device(s) requesting content from the cache are biased to lie closer to the cluster center. This means that in addition to the D2D link of interest, the intra-cluster interfering links may be more likely to have a transmitter or receiver closer to the cluster center. Incorporating such a behavior in the original model of Section~\ref{sec:sysmod} will require fixing the indices of the interfering devices in each cluster (relative to the cluster centers), as we did for the serving and receiving devices in the $k$-Tx and $\ell$-Rx cases. While this is certainly doable in principle, the final expressions will be prohibitively complex due to the dependence amongst all the distances involved in the coverage probability evaluation. For instance, the intra-cluster distances will have to be jointly handled through their joint distribution that will be evaluated using order statistics on the same lines as the serving distances in $k$-Tx and $\ell$-Rx  cases. Deconditioning on such joint distributions will result in multi-fold integrals that are not easy to evaluate. 

Therefore, to study the effect of this biasing of potential transmitters and receivers towards the cluster center, we propose a generative model in which the device locations follow a {\em double-variance Thomas cluster process}, where each cluster consists of a denser and a sparser subcluster. As discussed in this section, selecting the locations of the transmitters or receivers uniformly at random from the denser subcluster allows us to model the above described biasing while maintaining tractability. The generative model is illustrated in \figref{Fig: NetTopowith2sigma}. More formal details about the proposed model are presented next.%

   \begin{figure}[t!]
\centering{
        \includegraphics[width=.8\linewidth]{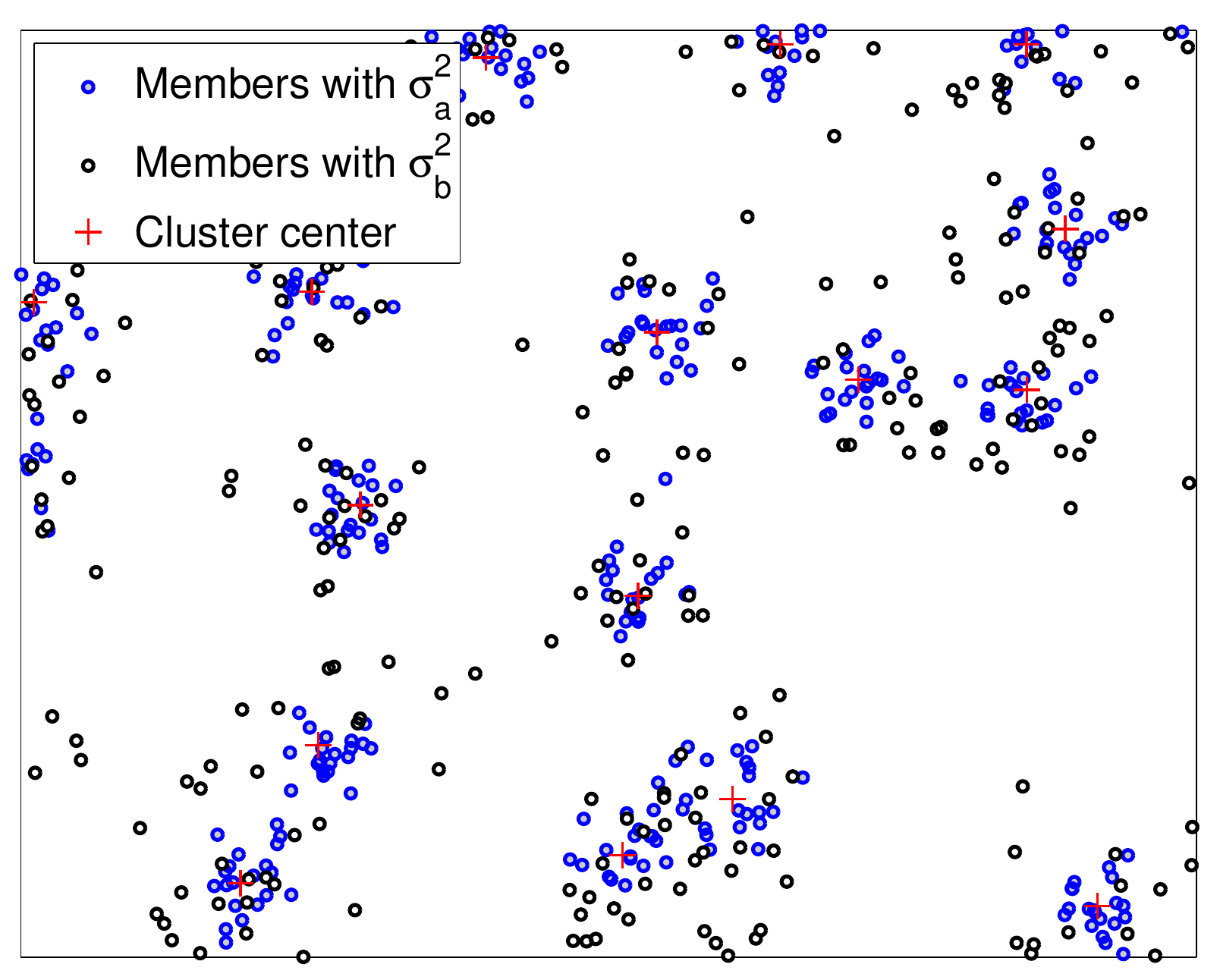}
              \caption{Illustration of the double-variance Thomas cluster model when cluster members (devices) are normally distributed around the cluster center with $\sigma_\nrma=10$ and $\sigma_\nrmb=30$.}
                \label{Fig: NetTopowith2sigma}
                }
\end{figure}
\subsection{System Setup and Channel Model} 

We model the clustered D2D network using a {\em double-variance Thomas cluster process} whose cluster centers are distributed according to a PPP $\Phi_\nrmc$ of density $\lambda_\nrmc$. A cluster around $x \in \Phi_\nrmc$ is formed of two subclusters, denser and sparser, of normally distributed devices with scattering variances $\sa^2$ and $\sb^2$, respectively. The analysis will be performed at a typical device, i.e., a device chosen uniformly at random from the either subclusters. This means that the D2D-Rx of interest will be normally distributed around the cluster center with variance $\sa^2$ or $\sb^2$. The simultaneously active transmitters of the two subclusters are denoted by $\Ax$ (denser) and $\Bx$ (sparser), where $|\Ax|$ and $|\Bx|$ are Poisson distributed with mean $\bar{m}_\nrma$ and $\bar{m}_\nrmb$, respectively. Since we want to bias the location of the D2D-Tx of interest towards the cluster center, we sample it uniformly at random from the denser subcluster $\Axx$. While the other case in which it is sampled from $\Bxx$ is not important for the current discussion, it can be handled exactly on the same lines. As was the case in the original model, the number of intra-cluster interfering devices in $\Axx$ is modeled as a PPP with mean $\bar{m}_\nrma-1$ to ensure that the average number of simultaneously active devices in this subcluster is $\bar{m}_\nrma$ (for consistency). Now assuming the serving device to be at $x_0+a_0 \in \Axx$, the intra-cluster interference at D2D-Rx of interest at the origin can be expressed as:

\small
\begin{align}\label{Eq: Intra_cluster_Int_typical_future}
      \intra=\  \sum_{\aa\in \Axx \setminus \txi} P_\nrmd \haa \|x_0+\aa\|^{-\alpha} 
              +\sum_{\bb\in \Bxx } P_\nrmd \hbb\|x_0+\bb\|^{-\alpha}.
\end{align}
\normalsize
Similarly, the inter-cluster interference can be expressed as: 
 \small
\begin{align}\label{Eq: Intera_cluster_Int_typical_future} 
\inter=\mathlarger{\sum_{x\in \phi_{\nrmc}\setminus x_0}}\Big[\sum_{\a \in \Ax} \ha \|x+\a\|^{-\alpha}+\sum_{\b \in \Bx} \hb \|x+\b\|^{-\alpha}\Big].
\end{align}
\normalsize
\subsection{Coverage Probability}
We first characterize the Laplace transform of inter-cluster and intra-cluster interference distributions in the following Lemmas.
\begin{lemma} 
\label{lem: Future Laplace intra} Assuming the serving device to be chosen uniformly at random from $\Axx$,
 the Laplace transform of the distribution of intra-cluster interference in~\eqref{Eq: Intra_cluster_Int_typical_future}, conditioned on $\nu_0=\|x_0\|$,  is given by $\calL_{\intra} (s |\nu_0)=$
\begin{multline}\label{Eq: laplace intra typical_futre}
 \exp\Big(-(\bar{m}_{\nrma}-1)\int_0^\infty \frac{s w_\nrma^{-\alpha}}{1+s w_\nrma^{-\alpha}}f_{W_{\nrma}}(w_\nrma|\nu_0)\nrmd w_\nrma\\
 -(\bar{m}_{\nrmb})\int_0^\infty \frac{s w_\nrmb^{-\alpha}}{1+s w_\nrmb^{-\alpha}}f_{W_{\nrmb}}(w_\nrmb|\nu_0)\nrmd w_\nrmb\Big), \quad \quad
\end{multline}
where $f_{W_k}(w_k| \nu_0)=\mathtt{Ricepdf}(w_k,\nu_0;\sigma_k^2)
,\:  k\in \{\nrma, \nrmb\}$.
\label{lem: lap intra typical_future}
\end{lemma}
\begin{proof}
See Appendix \ref{proof: lemma Laplace intra future}.
\end{proof}    
\begin{lemma}
\label{lem: Future Laplace inter}
Laplace transform of the distribution of  inter-cluster interference at D2D-Rx of interest in~\eqref{Eq: Intera_cluster_Int_typical_future} is given by $\calL_{\inter} (s )=$
\small
\begin{multline}\label{Eq: Lap_Inter_future}
\exp\Big(-2 \pi \lambda_\nrmc\int_0^\infty \Big(1-\exp\Big(-\ma \int_0^\infty \frac{s u_\nrma^{-\alpha}}{1+s u_\nrma^{-\alpha}} f_{U_{\nrma}}(u_\nrma|\nu)
\nrmd u_\nrma \\-\mb \int_0^{\infty}\frac{s u_\nrmb^{-\alpha}}{1+s u_\nrmb^{-\alpha}} f_{U_\nrmb}(u_\nrmb|\nu)\nrmd u_\nrmb\Big)\nu \nrmd \nu\Big)\Big),
\end{multline}
\normalsize
where $f_{U_k}(u_k| \nu)=\mathtt{Ricepdf}(u_k,\nu_0;\sigma_k^2), \:\: k\in\{\nrma, \nrmb\}$.
\end{lemma}
\begin{proof}
See Appendix \ref{proof : lap inter future}.
\end{proof}
Coverage probability when the content of interest is available at a device chosen uniformly at random from $\Axx$ is given by the following Theorem.
\begin{theorem}
Using the expression for Laplace transform of the intra-cluster interference distribution in \eqref{Eq: laplace intra typical_futre} and the inter-cluster interference in \eqref{Eq: Lap_Inter_future}, the coverage probability at a randomly chosen device from the double-variance model is
\small
\begin{align}\label{eq: coverage future}
 \pc=  \int_0^\infty  \int_0^\infty &\ncalL_{\inter}(\T r^{\alpha})\ncalL_{\intra}(\T r^{\alpha}|\nu_0) 
  f_R(r| \nu_0) f_{V_0}(\nu_0) \nrmd \,r \nrmd \nu_0,
\end{align}
\normalsize
 \begin{align}\label{Eq: typical serving distance future}
&\text{with,} \quad f_R(r| \nu_0)=\mathtt{Ricepdf}(r,\nu_0;\sigma_\nrma^2)
%
\end{align}
where $ f_{V_0}(\nu_0)=\frac{\nu_0}{   \sigma_k^2}\exp(-\frac{\nu_0^2}{2 \sigma_k^2})$. Here, $\sigma_k=\sigma_\nrma$ if D2D-Rx of interest  is located at denser cluster and  $\sigma_k=\sigma_\nrmb$ otherwise.
\label{thm: coverage new generative}
\end{theorem}
\begin{proof}
The proof follows on the same lines as Theorem \ref{Thm: Coverage case1 } with slight difference in the distance distributions. By definition, coverage probability is
\begin{equation}
\pc=\nbbE_{V_0} \nbbE_R \left[\nbbE\left[\exp\left(- \T r^\alpha (\inter+\intra) \right) \Big| R, V_0 \right]\right],
\end{equation}
where the PDF of  serving distance $r=\|x_0+a_0\|$,  conditioned on $\nu_0=\|x_0\|$,  follows Rician distribution \cite{MehrnazD2D1}. 
Here, D2D-Rx of interest is a randomly chosen device, which is sampled from a Gaussian distribution in $\nbbR^2$ with scattering variance $\sigma_\nrma$ ($\sigma_\nrmb$) if D2D-Rx of interest belongs to the denser (sparser) cluster. Hence $\nu_0=\|x_0\|$ follows Rayleigh distribution.
\end{proof}
It is worth highlighting that  the overall  performance of the \emph{double-variance} process will depend upon the following features: i)  {\em serving link distance:} it decreases when D2D-Tx of interest or D2D-Rx of  interest are located in the  denser subcluster, ii) {\em inter-cluster  interference:} it  decreases with the increase of the number of simultaneously active D2D-Txs in the denser subcluster compared to the sparser subcluster (keeping the total same), and iii) {\em intra-cluster interference:} it increases with the increase in the number of simultaneously active D2D-Txs in the denser subcluster. Here, the first two features, i.e., decreasing serving link distance and  inter-cluster interference, improves coverage probability while increasing intra-cluster interference degrades the coverage.
\section{Results and Discussion} \label{sec:NumResults}
\subsection{Numerical Results}
\subsubsection{Validation of results}
  \begin{figure}[t!]
\centering{
        \includegraphics[width=.82\linewidth]{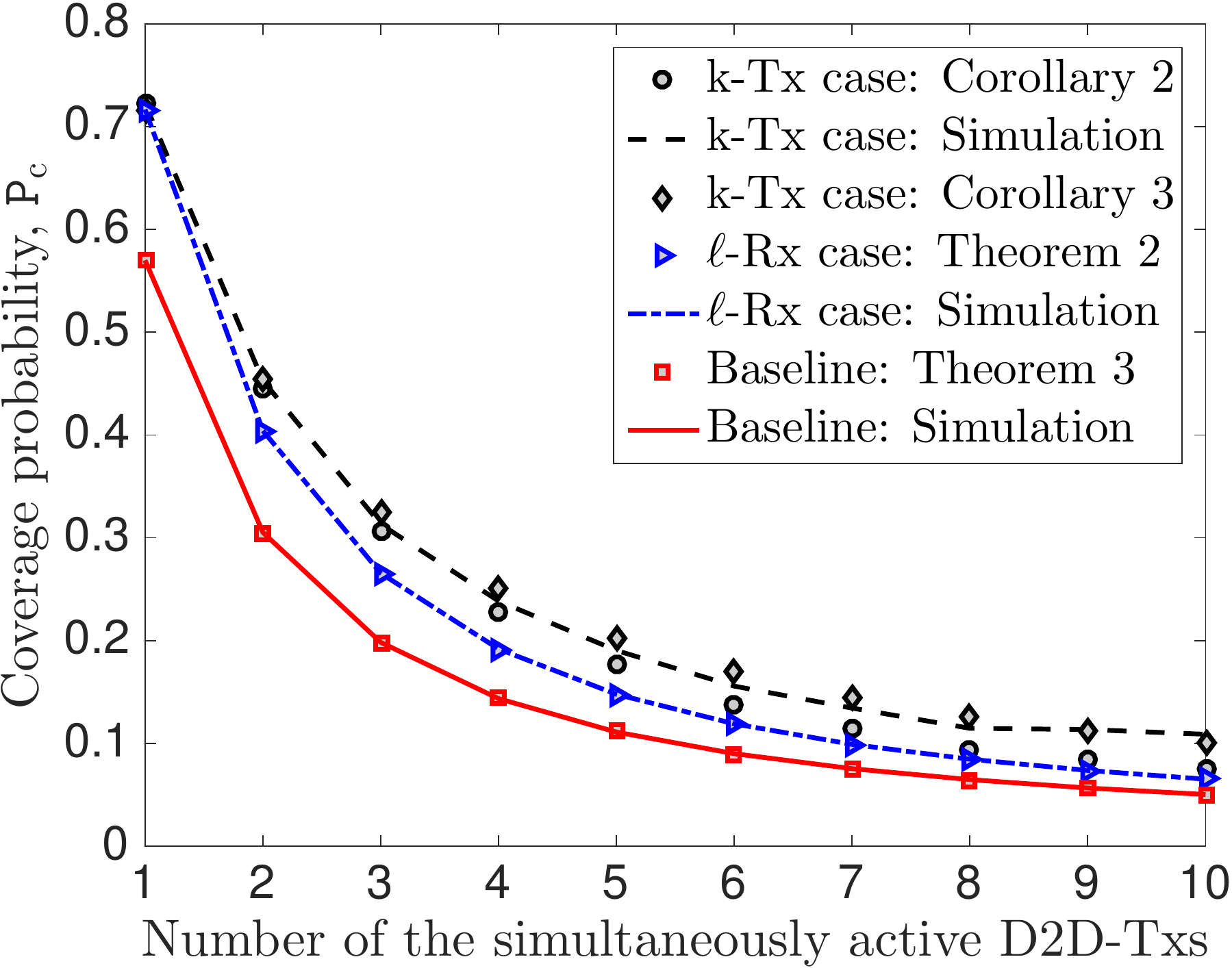}
              \caption{Coverage probability versus number of simultaneously
active D2D-Txs when $\sigma_\nrma = 30$,  $\lambda_\nrmc =50$ clusters $/$ km$^2$, $N_{\tt t}=N_{\tt r}=40$.
}
                \label{Fig: Validation}
                }
\end{figure}

In this section, we validate the accuracy of the analytical results, and tightness of the approximations by means of simulations. 
In all the simulations, the locations of cluster centers are a realization of a PPP and devices are normally  scattered around them.
 For this setup,
we set the $\sir$ threshold, $\beta$, as $0$ dB, path-loss exponent, $\alpha$ as $4$, and study the coverage probability for the three cases. While the easy-to-compute exact results for the $\ell$-Rx and baseline cases, given by Theorems~\ref{Thm: Coverage case2 } and~\ref{Thm: Coverage case3 }, are shown to match the simulations exactly, thus validating the analysis, the approximations for $k$-Tx case given by Corollaries \ref{corr: approx Coverage case1 } and \ref{corr: app coverage case 1} are also shown to be fairly tight. Although the exact expression for $k$-Tx case, given by Theorem~\ref{Thm: Coverage case1 }, is not straightforward to compute numerically, the tightness of approximation given by Corollary \ref{corr: approx Coverage case1 } means that it can be used as the proxy for the exact result. The results also show that the $k$-Tx and $\ell$-Rx cases lead to higher coverage probability than the baseline case. The difference in performance will be even more prominent in the $\ase$ that will be discussed in the next subsection.

   
\subsubsection{Performance comparison across three cases}
 \begin{figure}[t!]
\centering{
        \includegraphics[width=.85\linewidth]{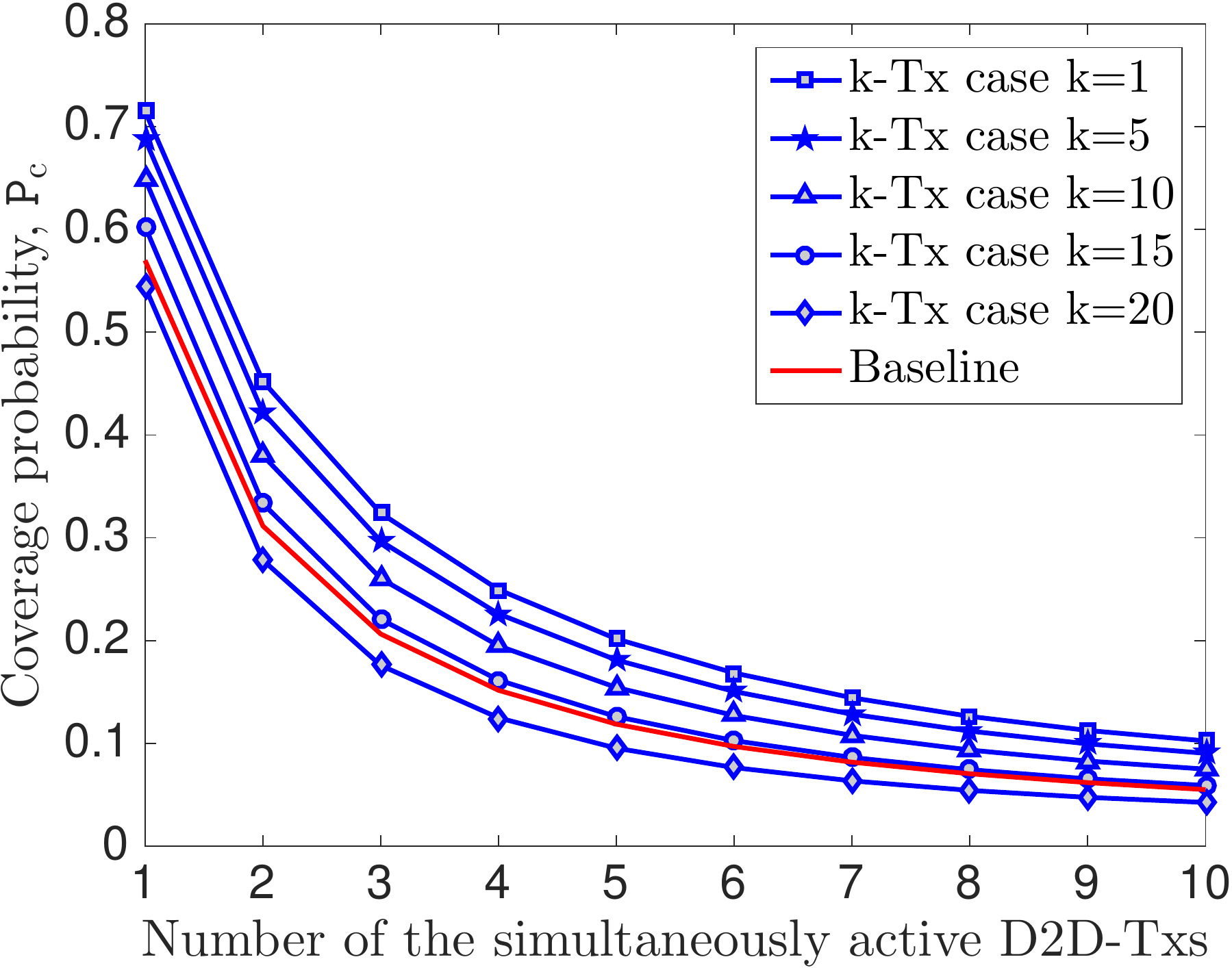}
              \caption{$k$-Tx case: Coverage probability versus number of simultaneously
active D2D-Txs when $\sigma_\nrma = 30$, $\lambda_\nrmc =50$ clusters $/$ km$^2$ and $N_{\tt t}=30$.
}
                \label{Fig: Pc_case1}
                }
\end{figure}
 \begin{figure}[t!]
\centering{
        \includegraphics[width=.85\linewidth]{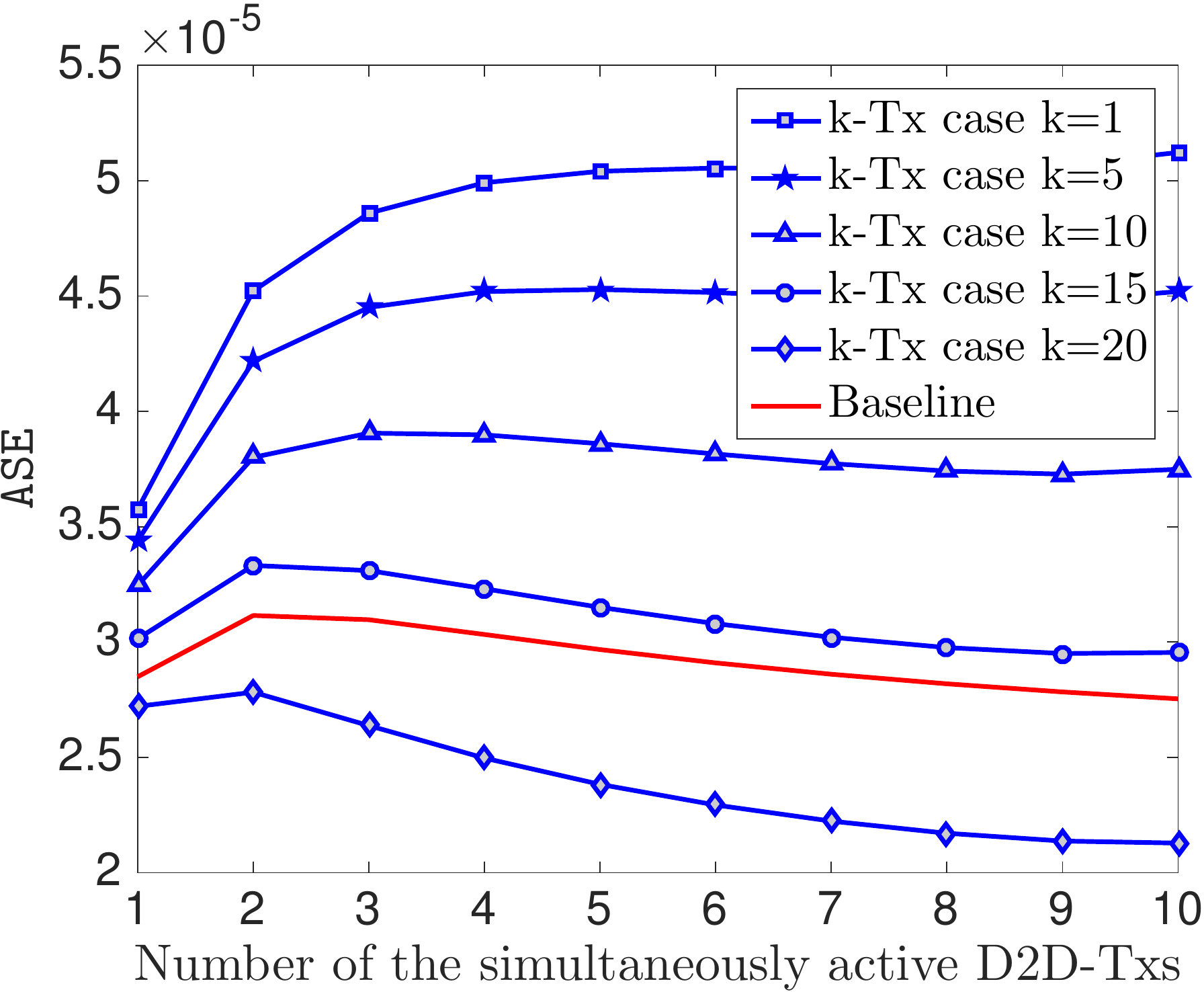}
              \caption{$k$-Tx case: $\ase$ versus number of simultaneously
active D2D-Txs when $\sigma_\nrma = 30$, $\lambda_\nrmc =50$ clusters $/$ km$^2$ and $N_{\tt t}=30$.
}
                \label{Fig: ASE_case1}
                }
\end{figure}
Recall that there is a clear trade-off between the optimal number of simultaneously active D2D-Txs and the resulting interference power. While increasing the  number of simultaneously active transmitters potentially increases $ \ase$, it comes at a price of an increased interference. 
 As shown in \figref{Fig: ASE_case1} for the $k$-Tx case, the optimal number of simultaneously active D2D-Txs increases with the decrease in distance from  the D2D-Tx of interest to the cluster center (i.e., decreasing $k$).   \figref{Fig: Pc_case1} and \figref{Fig: ASE_case1} show that the 
 coverage probability and $\ase$ are optimum when the content of interest for the D2D-Rx of interest is available at the closest device to the cluster center. The results also show that biasing the content of interest for the D2D-Rx towards the cluster center leads to a significant performance improvement in both coverage probability and $\ase$ compared to the baseline case. On the contrary, it can be seen that both coverage and $\ase$ in $k$-Tx case may become worse than the baseline case when the content is cached far from the cluster center (e.g., $k=20$ case in \figref{Fig: Pc_case1} and \figref{Fig: ASE_case1}). 

Similar trends are observed in \figref{Fig: Pc_case2} and \figref{Fig: ASE_case2} for the $\ell$-Rx case. In particular, the results show that the coverage and $\ase$ increase as the distance from D2D-Rx of interest to the cluster center is reduced. The impact of $\ell$ on the results, especially on the coverage probability, is however not as prominent as it was in the $k$-Tx case. This is due to the fact that while the serving distance reduces with decreasing $k$, the distances to the intra-cluster interfering devices also decrease in general, thus leading to a higher intra-cluster interference.


 \begin{figure}[t!]
\centering{
        \includegraphics[width=.85\linewidth]{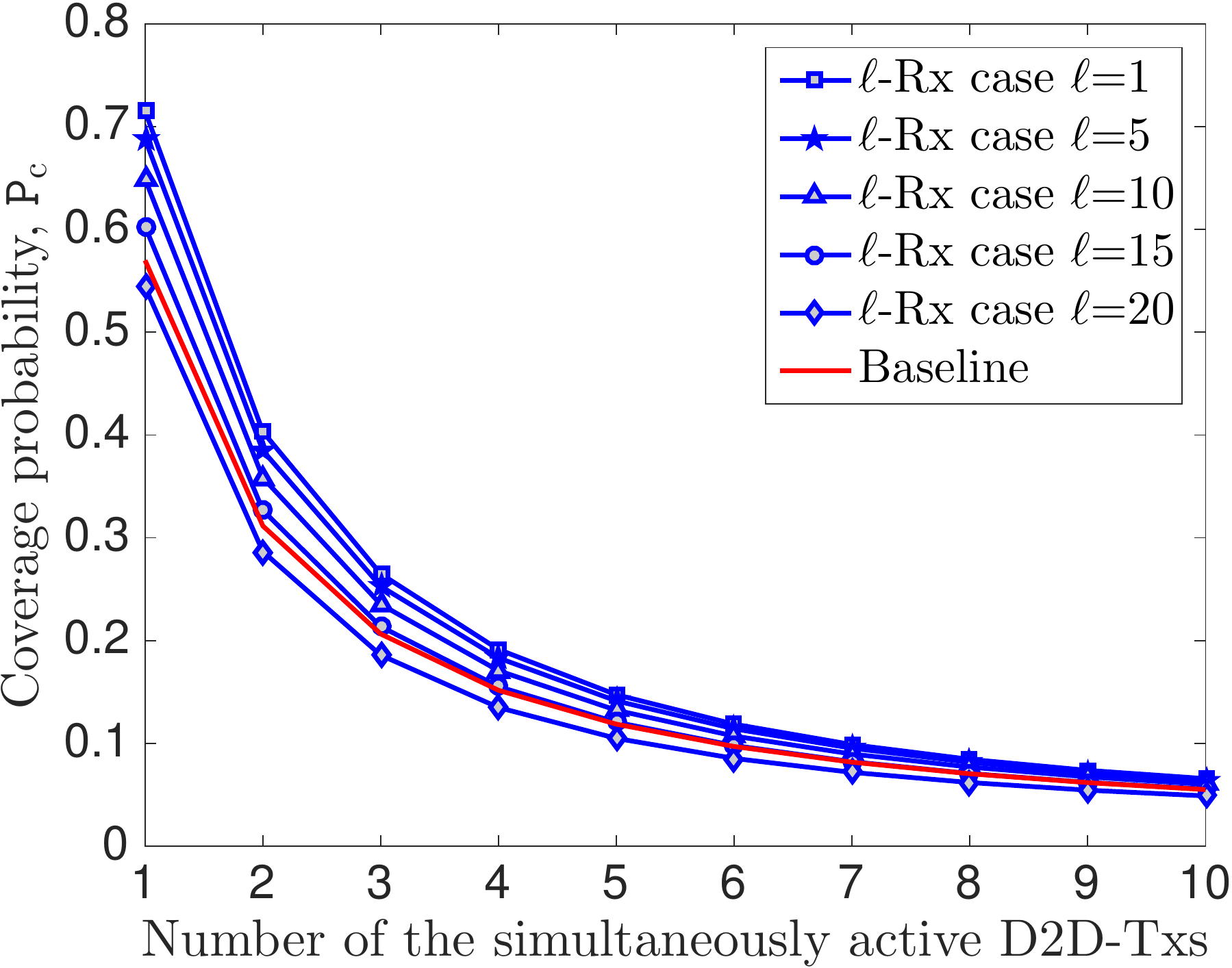}
              \caption{$\ell$-Rx case: Coverage probability versus number of simultaneously
active D2D-Txs when $\sigma_\nrma = 30$, $\lambda_\nrmc =50$ clusters $/$ km$^2$ and $N_{\tt r}=30$.
}
                \label{Fig: Pc_case2}
                }
\end{figure}

 \begin{figure}[t!]
\centering{
        \includegraphics[width=.85\linewidth]{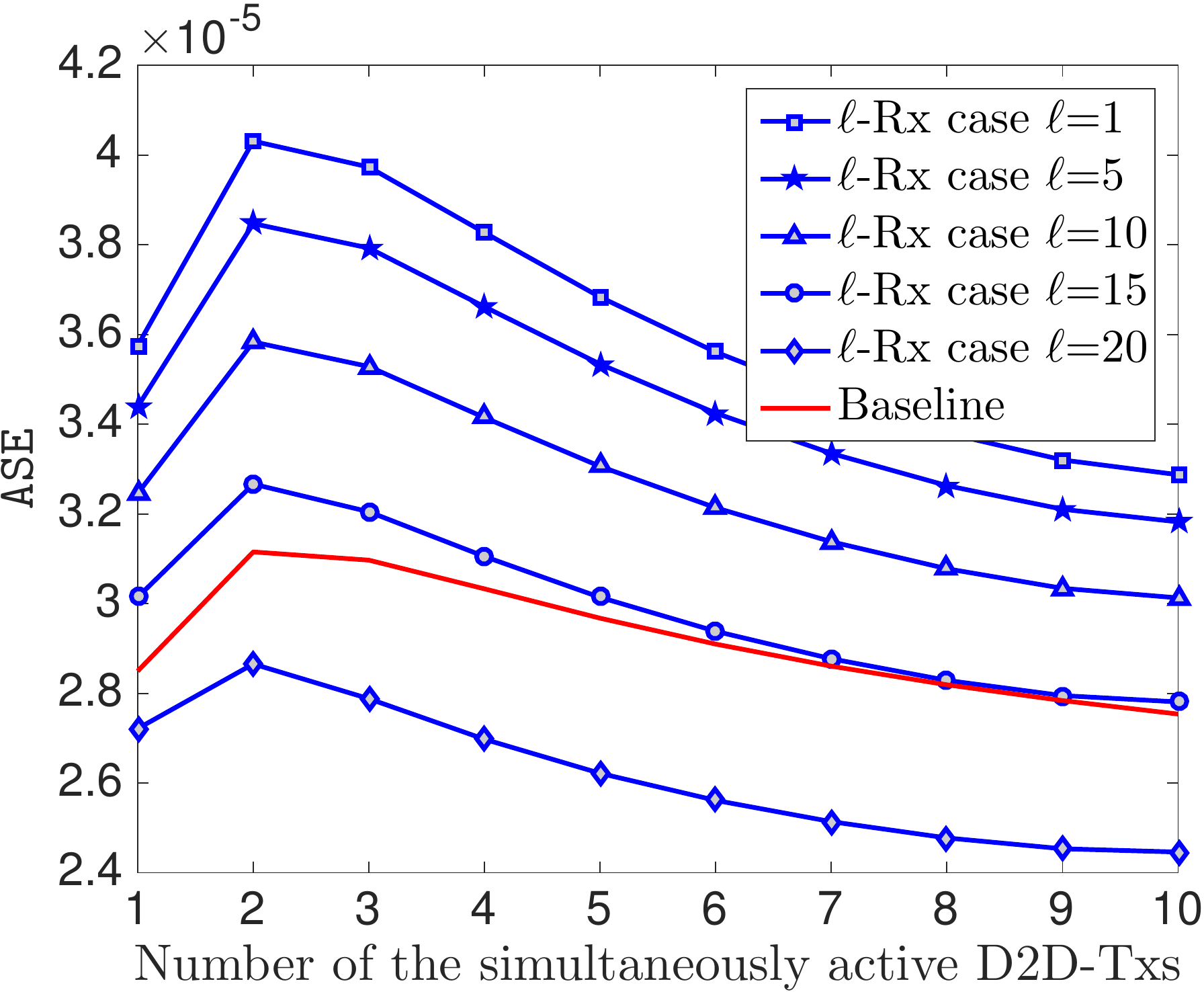}
              \caption{$\ell$-Rx case: $\ase$ versus number of simultaneously
active D2D-Txs when $\sigma_\nrma = 30$, $\lambda_\nrmc =50$ clusters $/$ km$^2$ and $N_{\tt r}=30$.
}
                \label{Fig: ASE_case2}
                }
\end{figure}
\chb{
\subsection{Applications of the Results to Total Hit Probability}
\label{subsec: Total Hit Probability}
In this section, we use the coverage probability results derived in this paper to study the D2D network performance in terms of the the {\em total hit probability}. We assume that the library of popular content for the representative cluster is known {\em a priori}. It is denoted by the set $\{c_1, c_2, ..., c_{\cal J }\}$, where the content is ordered in terms of decreasing popularity, which means $c_1$ denotes the most popular content. As is usually the case, we assume that the content popularity follows Zipf's distribution, i.e., the probability that the content $c_{j}$ is requested by the 
D2D-Rx of interest is ${\tt P}_{R_j}= \frac{j^{-\gamma}}{\sum_{j=1}^{\cal J}j^{-\gamma}}$,
where $\gamma$ is Zipf's parameter and ${\cal J}$ is total number of files~\cite{cha2007tube}. Note that the arguments presented in this section are not specific to Zipf distribution and can be easily extended to any given popularity distribution. The total hit probability can now be defined as the probability that the D2D-Rx of interest is able to successfully download its content of interest, which in turn depends upon two events: (i) this content of interest is available within the cluster, and (ii) the D2D-Rx of interest is in the coverage of the device that has this content (i.e., $\sir \ge \T$). Mathematical definition of total hit probability will be provided shortly.


Due to the limited storage capacities, each device in general cannot cache the whole library. For notational simplicity, we assume that each device caches exactly one content and the number of popular contents is greater than the total number of devices, i.e., ${\cal J}\ge N_{\tt t}$. Given that there are $ {N_{\tt t}}$ transmitting devices in each cluster, each device caches one of the ${N_{\tt t}}$ most popular contents \cite{6787081}. To evaluate the total hit probability, it is possible to have various cache {\em placement} and cache {\em gathering} policies. Due to space limitation, we confine our analysis to the following two strategies that directly build on the coverage probability results derived in this paper. For both these strategies, the D2D-Rx of interest is assumed to be chosen uniformly at random from the representative cluster, i.e., we confine to the $k$-Tx and baseline cases.

 \subsubsection{Uniform content placement}   
 \label{subsubsec: Uniform content placement}
In this setup, we assume that the popular contents are uniformly distributed inside the cluster. Recall that the coverage probability when the file of interest is available inside the cluster  uniformly at random is denoted  by ${\tt P}_{\rm c}^{\tt B}$ (baseline case). The total hit probability is
\begin{align}
{\tt P}_{\rm hit}&=\sum_{j=1}^{ N_{\rm t}}  {\tt P}_{R_j} {\tt P}_{\rm c}^{\tt B}
\end{align}
where ${\tt P}_{\rm c}^{\tt B}$ is the coverage probability given by \eqref{eq: coverage case 3}.

\subsubsection{ Cluster-centric content placement} \label{subsubsec: Cluster centric content placement}
Based on the intuition provided by Lemma  \ref{lem: Optimal content placement2} for the cluster-centric content placement, the most popular content should be placed at the transmitting device closest to the cluster center. This means that in our setup, $c_1$ should be placed at the device closest to the cluster center, and $c_2$ should be placed at the second closest device to the cluster center and so on. Hence the caching probability of the content $c_j$  at the $k^{th}$ closest transmitting device to the cluster center is ($j, k <N_{\tt t}$):
\begin{align*} 
b_{j,k}=\left\{
 \begin{array}{cc}
1 & j=k,\\
 0, & {\tt otherwise}.
 \end{array}\right.
 \end{align*}
Recall that the coverage probability when the randomly chosen D2D-Rx of interest connects to the $k^{th}$ closest transmitting device to the cluster center ($k$-Tx case) was denoted by ${{\tt P}_{c_k}^{\rm Tx}}$. Hence, the total hit probability can be expressed as
\begin{equation}
{\tt P}_{\tt hit}= 
\sum_{j=1}^{N_{\tt t}}  {\tt P}_{R_j}  {{\tt P}_{c_j}^{\rm Tx}},
\end{equation}
where ${{\tt P}_{{\rm c}_k}^{\rm Tx}}$ is the coverage probability given by \eqref{Eq: Pc case 1}.  

We now plot the total hit probability results for the two cases in  \figref{FigR: Hit prob J_40}. As expected, the total hit probability is significantly higher in the cluster-centric content placement case. While the shape parameter of Zipf distribution, $\gamma$, does not impact the results in the uniform content placement case, increasing its value improves the hit probability in the cluster-centric content placement case. This is because with increasing $\gamma$, the most popular content is requested more often and since it is cached closer to the cluster center, the D2D-Rx of interest connects with the devices located closer to the cluster center more often. Since the coverage probability in the $k$-Tx case for lower values of $k$ is significantly higher than for higher values of $k$, this improves the overall hit probability.

}
%
%
\begin{figure}
\centering{
 \includegraphics[width=.85\linewidth]{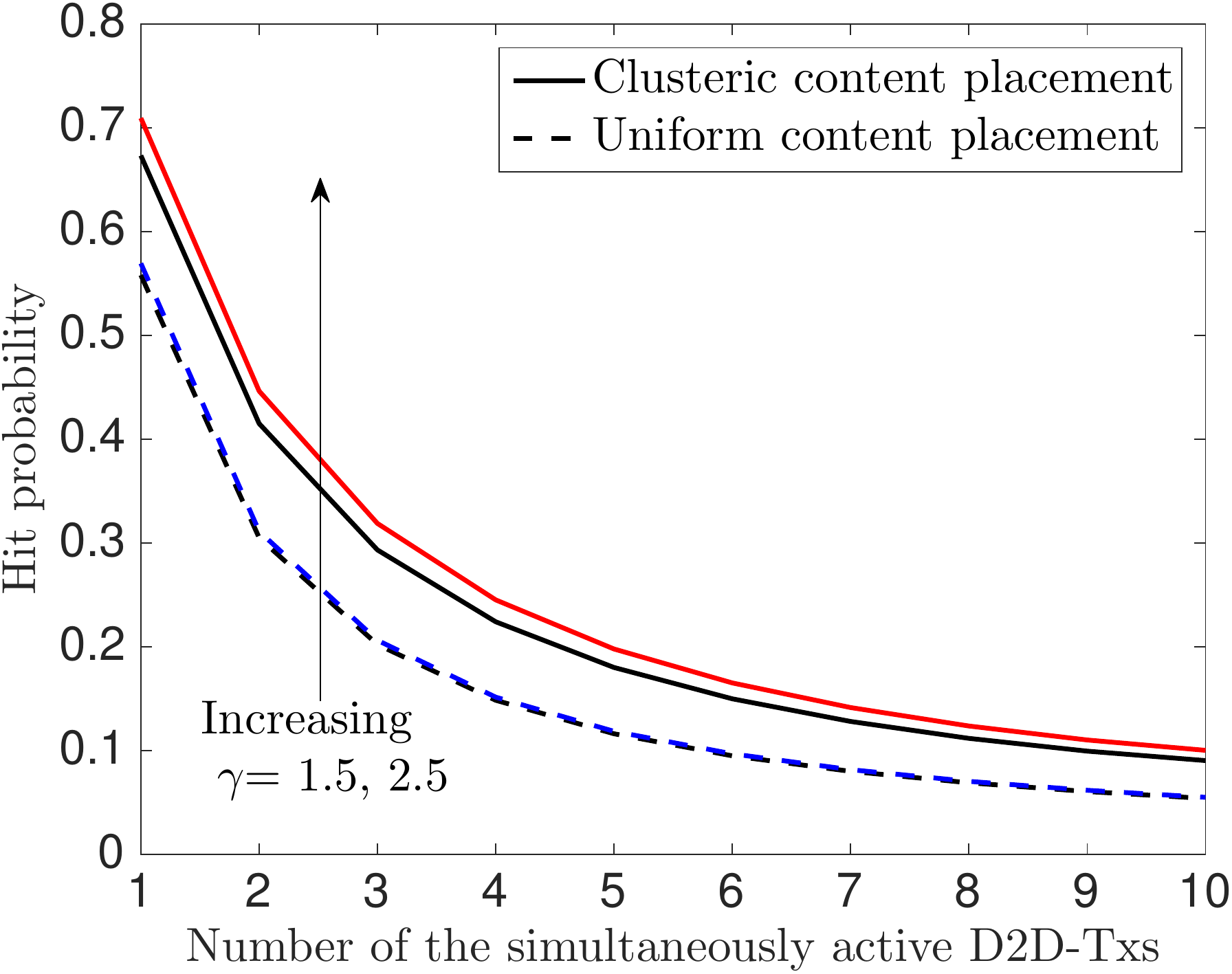}
  \caption{ Total hit probability versus number of simultaneously
active D2D-Txs when $\sigma_\nrma = 30$, $\lambda_\nrmc =50$ clusters $/$ km$^2$, $N_{\tt t}=30$, and ${\cal J}=40$.}
         \label{FigR: Hit prob J_40}
}
\end{figure}

\subsection{Performance of the  New Generative  Cluster-Centric Model} 
\label{subset: Performance of the  New Generative  Cluster-Centric Mode}
\chr{In this subsection, we study the coverage probability in the double-variance model of Section~\ref{Sec:Future of D2D network} as a function of the number of simultaneously active transmitters $\bar{m}_\nrma + \bar{m}_\nrmb$. The results are presented in \figref{Fig: Pc_future}. For each plot, we fix either $\bar{m}_\nrma$ or $\bar{m}_\nrmb$ and vary the other such that the sum is equal to the value on the x-axis. Recall that the analysis for this model was performed under the assumption that D2D-Tx of interest is sampled uniformly at random from the denser subcluster $\Axx$. It was stated that this assumption will lead to a better performance. This can be validated by noticing that the coverage probability for the case $\bar{m}_\nrma=0$ (the one where both the D2D-Tx and D2D-Rx of interest are in the sparser subcluster) is lower than all the cases in which the D2D-Tx of interest lies in $\Axx$.
Besides, the coverage probability when D2D-Rx of interest is located in the sparser subcluster is higher than the special case of $\bar{m}_\nrma=0$ discussed above, and lower than the other extreme in which both the D2D-Tx and D2D-Rx of interest are in the denser subcluster ($\bar{m}_\nrmb=0$). For all the cases, we can observe that the coverage probability is higher when the number of simultaneously active D2D-Txs in the sparser subcluster is higher (keeping $\bar{m}_\nrma+\bar{m}_\nrmb$ the same). This is because the active D2D-Txs in sparser subcluster cause less intra-cluster interference compared to when they lie in the denser subcluster. }

 \begin{figure}[t!]
\centering{
        \includegraphics[ width=.85\linewidth]{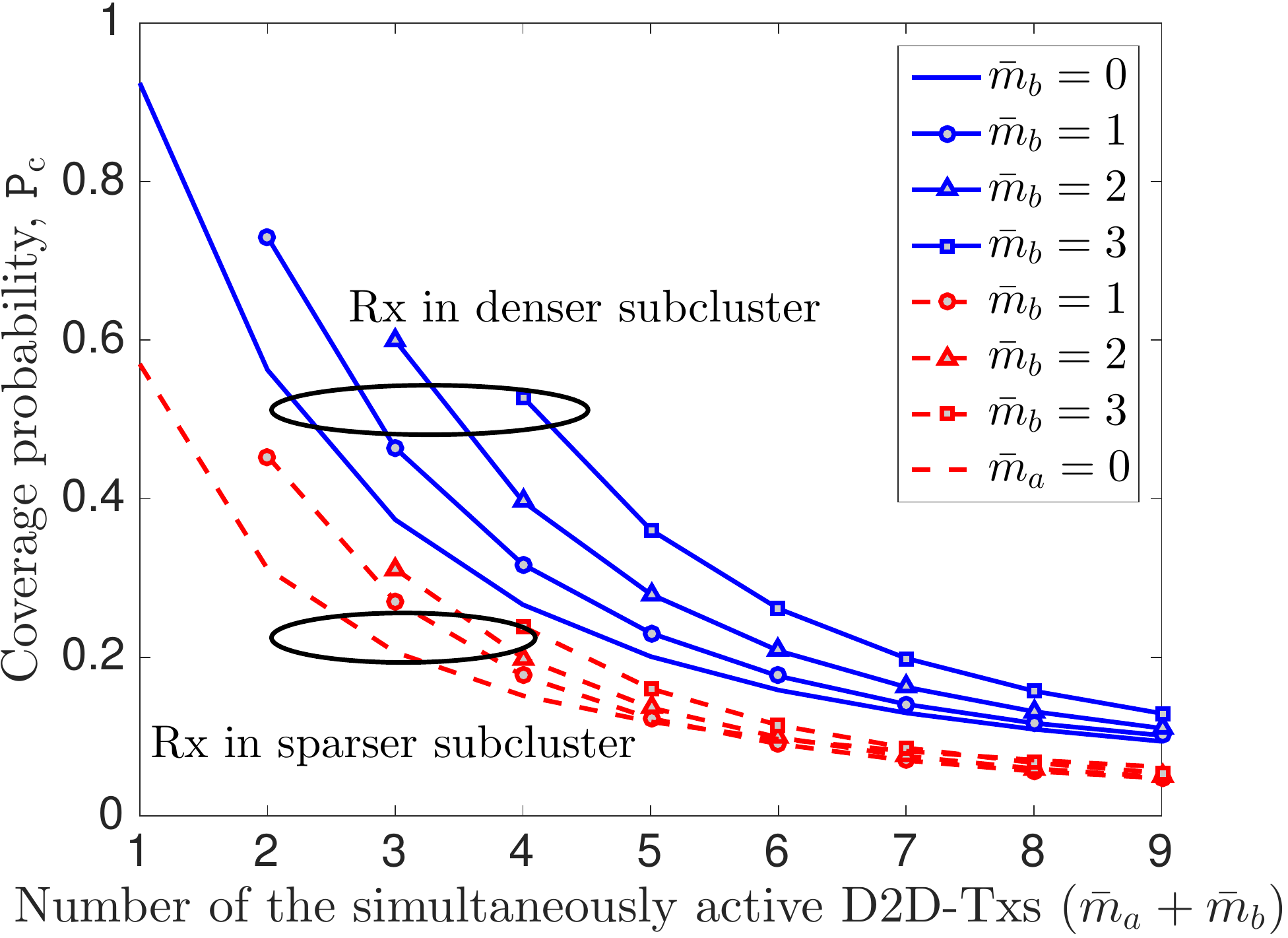}
              \caption{ Coverage probability versus number of simultaneously
active D2D-Txs when  and $\lambda_\nrmc =50$ clusters $/$ km$^2$, $\sa = 10$, and $\sb=30$.
}
                \label{Fig: Pc_future}
                }
\end{figure}
\section{Conclusion}
In this paper, we developed a realistic framework for the modeling and analysis of cache-enabled D2D networks. By modeling the D2D network as a Poisson cluster process, we focused on the performance analysis of the {\em cluster-centric} content placement policies where the content is placed in a cluster such that the collective performance of all the devices is improved. \chb{In particular, we defined and explored following two general cases where the location of the D2D-Rx of interest or the device that has cached its content of interest is parameterized in terms of its location relative to the cluster center: (i) $k$-Tx case: the serving device is the 
the $k^{th}$ closest device to the cluster center, and (ii) $\ell$-Rx case: the receiving device is the $\ell^{th}$ closest device to the cluster center. Using tools from stochastic geometry, we derived the coverage and  $\ase$ for these cases and compared them with the baseline case where both the D2D-Rx of interest and its serving device are chosen independently and uniformly at random from the same cluster.} The results concretely demonstrated that the network performance can be significantly enhanced if either the D2D-Tx of interest or the D2D-Rx of interest lie close to the cluster center. Based on this observation, we proposed and analyzed a new generative {\em double-variance Thomas cluster process} model that allowed us to tractably capture the fact that more intra-cluster interfering devices may lie closer to the cluster center. \chb{Several system design guidelines for content placement, including insights into the effect of realistic content placement policies on hit probabilities, have been provided.}
  
This work has many extensions. From the caching perspective, it is important to incorporate content popularity  distribution, social  interaction between devices and cache constraints, such as the memory of the caching devices. From the D2D network perspective, it is important to extend the analysis to the in-band case where the cellular and D2D transmissions share the same spectrum~\cite{AfsDhiC2015a}. From the cluster point process perspective, it is important to extend the analysis to more general classes of cluster processes. From the modeling perspective, extensive measurement campaigns must be carried out to understand the statistics of real-world clusters, such as the ones formed in the coffee shops, libraries, and restaurants.

\appendix
\subsection{Proof of Lemma \ref{lem: Optimal content placement2}} 
\label{proof: Optimal content placement}
Recall that the set of possible transmitting devices inside a representative cluster is denoted by $\ncalN^{x_0}_{\rm t} \in \{1,2,...,N_t\}$. Assuming the file of interest is located at $x_0+s_k$, we have
\begin{align*}
k^* =&\argmax_{ k\in \ncalN^{x_0}_{\rm t}}\E[\nb1 \{\sir(\|x_0+s_k\|)>\T\}]\\
\stackrel{(a)}{=}&\argmax_{ k\in \ncalN^{x_0}_{\rm t}} \E\left[\frac{ \htxi \|x_0+s_k\|^{-\alpha}}{I_T-\htxi \|x_0+s_k\|^{-\alpha}}>\T \right]\\
=&\argmax_{ k\in \ncalN^{x_0}_{\rm t}} \E\left[\frac{  \|x_0+s_k\|^{-\alpha}}{\frac{1}{\htxi}I_T-  \|x_0+s_k\|^{-\alpha}}>\T\right]\\
\stackrel{}{=}& \argmax_{ k\in \ncalN^{x_0}_{\rm t}} \E\left[X_k > \T \right]
\end{align*}
where $I_T$ in $(a)$ is the total received power at the D2D-Rx of interest from all the transmitters in the network, and $X_k$ is defined as the received $\sir$ from the $k^{th}$ closest serving device for the ease of notation. Since $I_T$ is not the function of $k$, $X_k \geq_{\rm st} X_j$ whenever $\|x_0+s_j\| \geq_{\rm st} \|x_0+s_k\|$, where $\geq_{\rm st}$ denotes first order stochastic dominance (or usual stochastic order). Note that since $x_0$ is  sampled from zero mean complex Gaussian random variable in  $\R^2$, the density function of $r=\|x_0+s_k\|$ conditioned on $t_k=\|s_k\|$  follows  Rician distribution with CDF $F_R(r| \y)=1-Q_1(\frac{\y}{\sigma_\nrma},\frac{r}{\sigma_\nrma})$, where $Q_1(\alpha,\beta)$ is the Marcum Q-function defined as $Q_1(\alpha,\beta) = \int_\beta^{\infty} y e^{-\frac{y^2 + \alpha^2}{2}} I_0(\alpha y) {\rm d} y$.
It turns out that $Q_1(\alpha,\beta)$ is monotonically increasing in $\alpha$ \cite[Property 11]{short2012computation}, which implies $Q_1(\frac{\y}{\sigma_\nrma},\frac{r}{\sigma_\nrma})$ is monotonically increasing in $t_k$, which implies $R(t_k)\le_{\rm st} R(t_k+1)$. This implies $\|x_0+s_j\| \geq_{\rm st} \|x_0+s_1\|$ $\forall j \neq 1$, which completes the proof. 
\subsection{Proof of Lemma \ref{lem: distance from intra-cluster device to typical device in case 1}}
\label{App: proof of distance from intra-cluster device to typical device in case 1}
Let the angle between the intra-cluster interfering device and the D2D-Rx of interest, as seen from the cluster center, be $\theta$. While the angle is uniformly distributed between $[0,2\pi]$, the ``direction'' is not important for the distance calculation, which means it suffices to assume it uniformly distributed between $[0,\pi]$. Now, the CDF of the distance between these two devices, 
conditioned on  $v_0=\|x_0\|$ and $t=\|a\|$ is 
\begin{align} 
F_{W}(w|\nu_0, t)&=\P [W<w|\nu_0, t]\\\notag
&\stackrel{(a)}{=}\P\left[\nu_0^2+t^2-2\nu_0 t \cos\theta\le w^2 |\nu_0, t\right]\\\notag
&=\P\left[\cos\theta\ge\frac{\nu_0^2+t^2-w^2}{2\nu_0 t}|\nu_0, t\right]\\\notag
&\stackrel{(b)}{=} \P\left[\theta <\cos^{-1}\left[\frac{\nu_0^2+t^2-w^2}{2\nu_0 t}\right]|\nu_0, t\right]\\\notag
&\stackrel{(c)}{=}\frac{1}{\pi} \cos^{-1}\left[\frac{\nu_0^2+t^2-w^2}{2\nu_0 t}|\nu_0, t\right],
\end{align}
where (a) follows from  the cosine law,  (b) follows from the fact that $\cos^{-1}$ is monotonically decreasing function, and (c) follows from the fact that $\theta \sim {\rm Unif}[0,\pi]$. Now, the conditional PDF of $f_W(w|\nu_0, t)$ is obtained by differentiating over $w$ as follows:
\begin{align*}
f_{W}(w|\nu_0, t)= \frac{1}{\pi} \frac{w/{\nu_0 t}}{\sqrt{1-\left(\frac{\nu_0^2+t^2-w^2}{2\nu_0 t}\right)^2}}, \:|\nu_0-t|<w<\nu_0+t.
\end{align*}
where $|.|$ denotes absolute value.
Using the fact that devices are normally scattered around cluster-center, the distance from intra-cluster devices to the cluster center, i.e., $t=\|a\|$ is Rayleigh distributed with parameter $\sigma_\nrma$, which implies that  the PDF of distance $t_{\rm in}$
of an intra-cluster interferer in the
set $ \Axx_{\rm in}$ i.e., $t_{\rm in}< t_k$ is truncated Rayleigh distribution 
 $$f_{T_{\rm in}}(t_{\rm in}|t_k)=f_T(t_{\rm in}|T_{\rm in }< t_k)= \frac{f_T(t_{\rm in})}{F_T(t_{k})}, \:\: t_{\rm in}< t_k,$$ 
where $f_T(.)$ and $F_T(.)$ are the PDF and CDF of Rayleigh distribution with parameter $\sigma_\nrma$ respectively. Similarly, the PDF of distance $t_{\rm out}$, where $t_{\rm out}> t_k$, is
 $$f_{T_{\rm out}}(t_{\rm out}|t_k)=f_T(t_{\rm out}|T_{\rm out }> t_k)= \frac{f_T(t_{\rm out})}{1-F_T(t_{k})}, \:\: t_{\rm out}> t_k.$$ 
 Using the same approach as \cite[Lemma 4]{MehrnazD2D1},  the  conditional i.i.d property of $w=\|x_0+a\|$, conditioned on  $v_0=\|x_0\|$ and $t=\|a\|$ can be formally proved.
 
\subsection{Proof of Lemma \ref{lem: Laplace exact case 1} }
\label{proof: lemma Laplace intra case 1}
Assuming that the file of interest is located at the $k^{th}$ closest transmitting device to the cluster center, we  divide the set of simultaneously active intra-cluster devices into three subsets: $\Axx=\{\Axx_{\rm in}, \txi, \Axx_{\rm out}\}$. Here $a_0$ denotes relative location of the serving device to the cluster center with distance $\|\txi\|=t_k$  away from it, where  $\Axx_{\rm in}=\big\{a\big| 
\|\aa\|<t_k \big\}$, and  $\Axx_{\rm out}=\big\{a\big|\|\aa\|>t_k \big\}$ are the set of devices  closer and further than serving device from cluster center, respectively. For this setup, the Laplace transform of distribution of intra-cluster interference  $\ncalL_{\intra}(s,t_k|\nu_0)$

\small
\begin{align}\notag
&\stackrel{(a)}{=}\E\Big[\exp\Big(-s \sum_{\aa \in \Axx \setminus \txi}  \haa\|x_0+ \aa\|^{-\alpha}\Big)\Big]\\\notag
&=\E\Big[\exp\Big(-s \Big(\sum_{\aa \in\Axx_{\rm in}}  \haa\|x_0+ \aa\|^{-\alpha}
\\&+\sum_{\aa \in\Axx_{\rm out}}  \haa\|x_0+ \aa\|^{-\alpha}\Big)\Big)\Big]\\\notag
&=\E\Big[\prod_{\aa \in\Axx_{\rm in}}\exp(-s\haa\|x_0+ \aa\|^{-\alpha} ) \\&\times\prod_{\aa \in\Axx_{\rm out}}\exp(-s\haa\|x_0+ \aa\|^{-\alpha} )\Big]\\\notag
&\stackrel{(b)}{=}\E\Big[\prod_{\aa \in\Axx_{\rm in}}\frac{1}{1+s\|x_0+ \aa\|^{-\alpha}}  \prod_{\aa \in\Axx_{\rm out}}\frac{1}{1+s\|x_0+ \aa\|^{-\alpha}}\Big] \\\notag
&\stackrel{(c)}{=}\sum_{n=0}^{{N_{\rm t}}-1}\sum_{l=0}^{g_{\rm m}} 
\Bigg(\underbrace{\int_{0}^{t_k} \int_{w_{\rm in}^{\rm L}}^{w_{\rm in}^{\rm U}} \frac{1}{1+s w^{-\alpha}} f_{W}(w|\nu_0,t_{\rm in})f_{T_{\rm in}}(t_{\rm in}|t_k) \nrmd w \nrmd t_{\rm in} }_{M_{\rm in}(s,t_k|\nu_0)} \Bigg)^l  \\\notag
&\times\Bigg(\underbrace{\int_{t_k}^{\infty} \int_{w_{\rm out}^{\rm L}}^{w_{\rm out}^{\rm U}} \frac{1}{1+s w^{-\alpha}} 
 f_{W}(w|\nu_0,t_{\rm out})f_{T_{\rm out}}(t_{\rm out}|t_k) \nrmd w \nrmd t_{\rm out}}_{M_{\rm out}(s,t_k|\nu_0)}\Bigg)^{n-l} \\\notag
 &\times \underbrace{\frac{{n \choose l} p^l(1-p)^{n-l}}{I({1-p};n-g_{\rm m},1+g_{\rm m})}}_{\P(L=l|L\le gm )}  \underbrace{\frac{(\bar{m}_\nrma-1)^n e^{-(\bar{m}_\nrma-1)}}{n! \xi}}_{\P(N=n|N\le {N_{\rm t}-1})}\\\notag
\end{align} 
\normalsize
with $w_{\rm in}^{\rm L}=|\nu-t_{\rm in}|$, $w_{\rm in}^{\rm U}=\nu_0+t_{\rm in}$, $w_{\rm out}^{\rm L}=|\nu_0-t_{\rm out}|$, $w_{\rm out}^{\rm U}=\nu_0+t_{\rm out}$,
 $p=\frac{k-1}{N_{\rm t}-1}$, $g_{\rm m}=\min(n,k-1)$, and $\xi=\sum_{j=0}^{N_{\rm t}-1}\frac{(\bar{m}_\nrma-1)^j e^{-(\bar{m}_\nrma-1)}}{j!}$.
Here (a) follows from definition of Laplace transform, (b) from the fact that $\haa\sim\exp(1)$, and (c)  from converting Cartesian to polar coordinates by using distance distribution given by Lemma \ref{lem: distance from intra-cluster device to typical device in case 1} along with conditional i.i.d. property of $f_{W}(w|\nu_0,t)$. Then, the result follows by  expectation over number of devices, where the  number of devices closer than serving device to the cluster center, i.e., $l$, is  binomial distributed conditioned on the total being less than $g_{\rm m}=\min(n,k-1)$. This condition is due to the fact that $l$ is always smaller than $k-1$ (since serving device is located at $k^{th}$ device in $k$-Tx case) and total number of active devices, i.e., $n$, where $n$  is Poisson distributed conditioned on total being less than $N_{\rm t}-1$. 
\subsection{Proof of Lemma  \ref{lem: lap intra typical}}
\label{proof: lemma Laplace intra}
The Laplace transform of distribution of intra-cluster interference can be derived as: $\ncalL_{\intra} (s|\nu_0)=\E\left[\exp\left(-s\intra\right)\right]$ 
\small
\begin{align}
\notag
\stackrel{(a)}{=}&\E\Big[\exp\Big(-s\sum_{\aa\in \Axx \setminus \txi}  \haa\|x_0+\aa\|^{-\alpha} 
\Big)\Big]\\\notag
\stackrel{}{=}&\E_{\Axx}\Big[\prod_{\aa\in \Axx \setminus \txi} \E_{\haa}\left[\exp\left(-s \haa\|x_0+\aa\|^{-\alpha} \right)\Big]\right] \\\notag
\stackrel{(b)}{=}&\E_{\Axx}\Big[\prod_{\aa\in \Axx \setminus \txi} \frac{1}{1+s \|x_0+\aa\|^{-\alpha}}  \Big] \\\notag
\stackrel{(c)}=&\sum_{n=0}^{N_{\tt t}-1} \Big(\int_{\R^2} \frac{1}{1+s \|x_0+\aa\|^{-\alpha}} f_A(\a) \nrmd \a \Big)^n 
\frac{(\bar{m}_\nrma-1)^n e^{-(\bar{m}_\nrma-1)}}{n! \xi}\\\notag
\stackrel{(d)}=&\sum_{n=0}^{N_{\tt t}-1} \Big(\int_0^\infty \frac{1}{1+s w^{-\alpha}}f_{W}(w|\nu_0) \nrmd w \Big)^n 
\frac{(\bar{m}_\nrma-1)^n e^{-(\bar{m}_\nrma-1)}}{n! \xi}\notag
  \end{align}
\normalsize
with $\xi=\sum_{j=0}^{N_{\tt t}-1}\frac{(\bar{m}_\nrma-1)^j e^{-(\bar{m}_\nrma-1)}}{j!}$, where (a) follows from the definition of Laplace transform,  (b) follows from the fact that $\haa$ is exponential distributed with mean unity, (c) follows from expectation over number of devices that is Poisson distributed conditioned on the total being less than $N_{\tt t}-1$ along with the fact that locations of devices conditioned on the location of cluster center, $x_0$, are i.i.d, and (d)
 follows from converting Cartesian to polar coordinates using the fact that $f_W(w|\nu_0)$, the density function of distance from interfering devices to the D2D-Rx of interest conditioned on $\nu_0=\|x_0\|$,  has Rician distribution given by   \eqref{Eq: Rice distribution}.
 Now under the assumption  $\bar{m}_\nrma \ll N_{\tt t}$, 
 the Laplace transform of intra-cluster  interference distribution can be approximated  as:
\begin{align*}
\simeq\frac{1}{\xi }      \exp\Big(-(\bar{m}_\nrma-1)\int_0^\infty \frac{s w^{-\alpha}}{1+s w^{-\alpha}}f_{W}(w|\nu_0) \nrmd w\Big).
\end{align*}
\subsection{Proof of Corollary \ref{corr: app coverage case 1}}
\label{proof:  Corollary app coverage case1}
Under Assumption \ref{Ass: Identical intra-cluster distances}, the correlation corresponding to the common distance $\nu_0=\|\mathbf{x_0}\|$ is ignored and hence the Laplace transform of the interference distribution can be approximated by Corollary \ref{cor: app lap intra}. Furthermore, the  density function of serving distance $r=\|\mathbf{x_0}+\mathbf{s_k}\|\in \R_+,$  needs to be evaluated only conditioned on the  
 distance of D2D-Tx of interest (i.e., $k^{th}$ closest device) to the cluster center $t_k=\|\mathbf{s_k}\|$. Now using the fact that $\mathbf{x_0}\in\R_+$ is zero mean complex Gaussian random variable, the density function of $\mathbf{z=x_0+s_k}\in \R^2$, where $\mathbf{z}=(z_1,z_2)$ conditioned on $\mathbf{s_k}=(s_{k_1},s_{k_2})$ (where $t_k=\sqrt{s_{k_1}^2+s_{k_2}^2}$) can be expressed as: 
  \begin{multline*}
 f_{\mathbf{Z}}(z_1,z_2|\mathbf{s_k})= \frac{1}{{2 \pi} \sigma_\nrma^2}\exp \left(-\frac{(\rx-s_{k_1})^2}{2 \sigma_\nrma^2} -\frac{(\ry-s_{k_2})^2}{2 \sigma_\nrma^2}\right).
 \label{eq: f_z}
 \end{multline*}
 \normalsize
Since, we are interested on distribution of $r=\|\textbf{z}\|$, we define $z_1= r \sin \theta$, and $z_2=r \cos \theta$, where $\theta=\arctan(\frac{z_1}{z_2})$. Now,  Jacobian matrix is used to convert  Cartesian coordinates to polar coordinates as follows:
\begin{align}
f_{R,\Theta} (r,\theta | \mathbf{s_k}) = f_\mathbf{Z}(z_1, z_2|\mathbf{s_k}) \times \left| \partial \left( \frac{z_1, z_2}{r,\theta} \right) \right|,
\end{align}
where $\partial \left( \frac{z_1, z_2}{r,\theta} \right) = \left| \begin{array}{cc} 
\frac{\partial z_1}{\partial r} & \frac{\partial z_1}{\partial \theta}\\
\frac{\partial z_2}{\partial r} & \frac{\partial z_2}{\partial \theta}
\end{array} \right| = r$, and hence joint distribution of $(R,\Theta)$ is $f_{R,\Theta} (r,\theta | \mathbf{s_k})=$ 
\begin{align*}
 &\frac{r}{{2 \pi} \sigma_\nrma^2}\exp \left(-\frac{(r \cos \theta-s_{k_1})^2}{2 \sigma_\nrma^2} -\frac{(r \sin \theta-s_{k_2})^2}{2 \sigma_\nrma^2}\right)\\
\stackrel{}{=}&  \frac{r}{ \sigma_\nrma^2} \exp \left( - \frac{r^2 + t_k^2}{2 \sigma_\nrma^2} \right) \frac{1}{2 \pi} \exp \left( \frac{r s_{k_1} \cos \theta + r s_{k_2}\sin \theta}{\sigma_\nrma^2} \right),
\end{align*}
Therefore,  the conditional marginal distribution of $R$ is 
\begin{align*}
&f_R(r| t_k) =\int_{0}^{2\pi} f_{R,\Theta} (r,\theta | \mathbf{s_k}) {\rm d} \theta =
\frac{r}{ \sigma_\nrma^2} \exp \left( - \frac{r^2 + t_k^2}{2 \sigma_\nrma^2} \right)\\ &\times\underbrace{\int_{0}^{2\pi} \frac{1}{2 \pi} \exp \left( \frac{r  s_{k_1}\cos \theta + r  s_{k_2} \sin \theta}{\sigma_\nrma^2} \right) {\rm d} \theta}_{I_0\left(\frac{r t_k}{\sigma_\nrma^2}\right)},
\end{align*}
where conditioning on $t_k=\|\mathbf{s_k}\|$ instead of $\mathbf{s_k}$ suffices. The rest of the proof follows on the same line as the proof of Theorem \ref{Thm: Coverage case1 }.

\subsection{Proof of Lemma \ref{lem: Future Laplace intra}}
\label{proof: lemma Laplace intra future}
Since the two sets  $\Axx$ and $\Bxx$ are independent, the Laplace transform of intra-cluster interference distribution, $\ncalL_{\intra} (s|\nu_0)=\E\left[\exp\left(-s\intra\right)\right]$, can be derived as follows:

\small
\begin{align}
\notag 
\stackrel{(a)}{=}&\E\Big[\exp\Big(-s\sum_{\aa\in \Axx \setminus \txi}  \haa\|x_0+\aa\|^{-\alpha}\\\notag
&+\sum_{\bb\in \Bxx }  \hbb\|x_0+\bb\|^{-\alpha}\Big)\Big]\\\notag
\stackrel{(b)}{=}&\E_{\Axx}\Big[\prod_{\aa\in \Axx \setminus \txi} \E_{\haa}\left[\exp\left(-s \haa\|x_0+\aa\|^{-\alpha} \right)\Big]\right] \\\notag 
&\times  \E_{\Bxx}\Big[\prod_{\bb\in \Bxx } \E_{\hbb}\left[\exp\left(-s \hbb\|x_0+\bb\|^{-\alpha} \right)\Big]\right]\\\notag
\stackrel{(c)}{=}&\E_{\Axx}\Big[\prod_{\aa\in \Axx \setminus \txi} \frac{1}{1+s \|x_0+\aa\|^{-\alpha}}  \Big]\\ & \times\E_{\Bxx}\Big[\prod_{\bb\in \Bxx } \frac{1}{1+s \|x_0+\bb\|^{-\alpha}}  \Big]\\\notag
\stackrel{(d)}{=}&\exp\Big(-(\bar{m}_{\nrma}-1) \Big(\int_{\R^2} \frac{\|\aa+x_0\|^{-\alpha}}{1+s \|x_0+\aa\|^{-\alpha}} f_A(\a) \nrmd \a \Big)\\\notag
& -\bar{m}_{\nrmb}\int_{\R^2}\frac{\|x_0+\bb\|^{-\alpha}}{1+s \|x_0+\bb\|^{-\alpha}} f_B(\bb) \nrmd \bb \Big)\\\notag
\stackrel{(e)}{=}&\exp\Big(-(\bar{m}_{\nrma}-1)\int_0^\infty \frac{s w_\nrma^{-\alpha}}{1+s w_\nrma^{-\alpha}}f_{W_{\nrma}}(w_\nrma|\nu_0)\nrmd w_\nrma \\\notag
& -\bar{m}_{\nrmb}\int_0^\infty \frac{s w_\nrmb^{-\alpha}}{1+s w_\nrmb^{-\alpha}}f_{W_{\nrmb}}(w_\nrmb|\nu_0)\nrmd w_\nrmb\Big)\notag
\end{align}
\normalsize
where (a) follows the definition of Laplace transform, (b) follows from the fact that the interference from set $\Axx$ and $\Bxx$ are independent,  (c) follows from the fact that $\haa$ and $\hbb$ are exponential random variables with mean unity, (d)  follows from probability generating functional (PGFL) of  Poisson distribution where $f_B(\b)=\frac{1}{2 \pi \sigma_\nrmb^2 }\exp\left(\frac{-\|\b\|^2}{2 \sigma_\nrmb^2}\right)$, and (e)  follows from converting from Cartesian to polar coordinates and some algebraic manipulation similar to the derivation of Rician distribution in the proof  of Corollary \ref{corr: app coverage case 1}.

\subsection{Proof of Lemma \ref{lem: Future Laplace inter}}
\label{proof : lap inter future}
Laplace transform of the  distribution intra-cluster interference at D2D-Rx of interest $\ncalL_{\inter}(s)$ is
 
\small
\begin{align*}\label{Eq: Lap_Intera_cluster_Intn}
            =&\E\Big[\exp(-s\mathlarger{\sum_{x\in \phi_{\nrmc}\setminus x_0}}\Big[\sum_{\a \in \Ax} \ha\|x+\a\|^{-\alpha}+\sum_{\b \in \Bx} \hb\|x+\b\|^{-\alpha}\Big])\Big]\\
\stackrel{(a)}  {=}&\E_{\phi_{\nrmc}}\Big[\prod_{x\in \phi_{\nrmc}\setminus x_0}\E_{\Ax} \Big[\prod_{\a \in \Ax}\E_{\ha}\left[\exp(-s \ha\|x+\a\|^{-\alpha})|x\right]\Big]
\\
&\E_{\Bx}\Big[\prod_{\b \in \Bx}\E_{\hb}\left[\exp(-s \hb\|x+\b\|^{-\alpha})|x]\right]\Big]\Big]\\
    \stackrel{(b)}{=}&\E_{\phi_{\nrmc}}\Big[\prod_{x\in \phi_{\nrmc}\setminus x_0}\E_{\Ax} \Big[\prod_{\a \in \Ax} \frac{1}{1+s\|x+\a\|^{-\alpha}}|x\Big]\\ 
    &\E_{\Bx}\Big[ \prod_{\b \in \Bx} \frac{1}{1+s\|x+\b\|^{-\alpha}}|x\Big]\Big],\\
  \stackrel{(c)} {=}&\exp\Big(-2 \pi \lambda_\nrmc  \int_0^\infty \Big(1-\exp\Big(-\ma \int_0^\infty \frac{s u_\nrma^{-\alpha}}{1+s u_\nrma^{-\alpha}} f_{U_{\nrma}}(u_\nrma|\nu)\nrmd u_\nrma \\
 & -\mb \int_0^{\infty}\frac{s u_\nrmb^{-\alpha}}{1+s u_\nrmb^{-\alpha}} f_{U_\nrmb}(u_\nrmb|v)\nrmd u_\nrmb\Big)\nu \nrmd \nu\Big)\Big)
 \end{align*} 
 \normalsize 

 where (a) follows from the fact that $\Ax$ and $\Bx$ conditioned on the location of cluster centers $\{x\}$ are independent,  (b) follows from the fact that $\ha$ and $\hb$ are independent exponential random variables with mean unity, and (c) follows from the PGFL of  Poisson distribution.
Note that Lemma \ref{Lem: Lap_Inter} is a special case of Lemma \ref{lem: Future Laplace inter} when $\bar{m}_\nrmb=0$.
 \section*{Acknowledgment}
 The authors would like to thank Surabhi Gaopande and SaiDhiraj Amuru  for helpful feedback.
\balance
\bibliographystyle{IEEEtran}
\bibliography{2_D2D2_V21.bbl}

\begin{thebibliography}{10}
\providecommand{\url}[1]{#1}
\csname url@samestyle\endcsname
\providecommand{\newblock}{\relax}
\providecommand{\bibinfo}[2]{#2}
\providecommand{\BIBentrySTDinterwordspacing}{\spaceskip=0pt\relax}
\providecommand{\BIBentryALTinterwordstretchfactor}{4}
\providecommand{\BIBentryALTinterwordspacing}{\spaceskip=\fontdimen2\font plus
\BIBentryALTinterwordstretchfactor\fontdimen3\font minus
  \fontdimen4\font\relax}
\providecommand{\BIBforeignlanguage}[2]{{%
\expandafter\ifx\csname l@#1\endcsname\relax
\typeout{** WARNING: IEEEtran.bst: No hyphenation pattern has been}%
\typeout{** loaded for the language `#1'. Using the pattern for}%
\typeout{** the default language instead.}%
\else
\language=\csname l@#1\endcsname
\fi
#2}}
\providecommand{\BIBdecl}{\relax}
\BIBdecl

\bibitem{AfsDhiC2015b}
M.~Afshang, H.~S. Dhillon, and P.~H.~J. Chong, ``Fundamentals of
  cluster-centric content placement in device-to-device networks,'' in {\em
  Proc. IEEE Globecom workshops}, San Diego, CA, Dec. 2015.

\bibitem{Cisco2015}
Cisco, ``Cisco visual networking index: Global mobile data traffic forecast
  update 2014{-} 2019 white paper,'' 2015.

\bibitem{andrews2014will}
J.~G. Andrews, S.~Buzzi, W.~Choi, S.~Hanly, A.~Lozano, A.~C. Soong, and J.~C.
  Zhang, ``What will {5G} be?'' \emph{IEEE Journal on Selected Areas in
  Communications}, vol.~32, no.~6, pp. 1065--1082, Jun. 2014.

\bibitem{golrezaei2013femtocaching}
N.~Golrezaei, A.~F. Molisch, A.~G. Dimakis, and G.~Caire, ``Femtocaching and
  device-to-device collaboration: A new architecture for wireless video
  distribution,'' \emph{IEEE Commun. Magazine}, vol.~51, no.~4, pp. 142--149,
  Apr. 2013.

\bibitem{boccardi2014five}
F.~Boccardi, R.~W. Heath~Jr., A.~Lozano, T.~Marzetta, and P.~Popovski, ``Five
  disruptive technology directions for {5G},'' \emph{IEEE Commun. Magazine},
  vol.~52, no.~2, pp. 74--80, Feb. 2014.

\bibitem{song2015wireless}
L.~Song, D.~Niyato, Z.~Han, and E.~Hossain, \emph{Wireless Device-to-Device
  Communications and Networks}.\hskip 1em plus 0.5em minus 0.4em\relax
  Cambridge University Press, 2015.

\bibitem{cha2007tube}
M.~Cha, H.~Kwak, P.~Rodriguez, Y.-Y. Ahn, and S.~Moon, ``{I Tube}, {You Tube},
  {Everybody Tubes}: analyzing the world's largest user generated content video
  system,'' in \emph{Proc., ACM Intl. Conf. on Special Interest Group on Data
  Commun. (SIGCOMM)}, 2007.

\bibitem{fast2005creating}
A.~Fast, D.~Jensen, and B.~N. Levine, ``Creating social networks to improve
  peer-to-peer networking,'' in \emph{Proc., ACM Intl. Conf. on Special
  Interest Group on Knowledge Discovery and Data Mining}, Aug. 2005.

\bibitem{tadrous2014joint}
J.~Tadrous, A.~Eryilmaz, and H.~El~Gamal, ``Joint pricing and proactive caching
  for data services: Global and user-centric approaches,'' in \emph{Proc., IEEE
  INFOCOM}, 2014.

\bibitem{bastug2014living}
E.~Bastug, M.~Bennis, and M.~Debbah, ``Living on the edge: The role of
  proactive caching in 5g wireless networks,'' \emph{IEEE Commun. Magazine},
  vol.~52, no.~8, pp. 82--89, 2014.

\bibitem{maddah2014fundamental}
M.~Maddah-Ali and U.~Niesen, ``Fundamental limits of caching,'' \emph{IEEE
  Trans. on Info. Theory}, vol.~60, no.~5, pp. 2856--2867, May 2014.

\bibitem{ahlehagh2012video}
H.~Ahlehagh and S.~Dey, ``Video-aware scheduling and caching in the radio
  access network,'' \emph{IEEE/ACM Trans. on Networking}, vol.~22, no.~5, pp.
  1444--1462, Oct. 2012.

\bibitem{shanmugam2013femtocaching}
K.~Shanmugam, N.~Golrezaei, A.~G. Dimakis, A.~F. Molisch, and G.~Caire,
  ``Femtocaching: {Wireless} content delivery through distributed caching
  helpers,'' \emph{IEEE Trans. on Info. Theory}, vol.~59, no.~12, pp.
  8402--8413, Dec. 2013.

\bibitem{gitzenis2013asymptotic}
S.~Gitzenis, G.~Paschos, and L.~Tassiulas, ``Asymptotic laws for joint content
  replication and delivery in wireless networks,'' \emph{IEEE Trans. on Info.
  Theory}, vol.~59, no.~5, pp. 2760--2776, 2013.

\bibitem{BlaszczyszynG14}
B.~Blaszczyszyn and A.~Giovanidis, ``Optimal geographic caching in cellular
  networks,'' in \emph{Proc., IEEE Intl. Conf. on Commun. (ICC)}, Jun. 2015.

\bibitem{molisch2014caching}
A.~F. Molisch, G.~Caire, D.~Ott, J.~R. Foerster, D.~Bethanabhotla, and M.~Ji,
  ``Caching eliminates the wireless bottleneck in video aware wireless
  networks,'' \emph{Advances in Electrical Engineering}, Nov. 2014.

\bibitem{ji2013wireless}
M.~Ji, G.~Caire, and A.~F. Molisch, ``Wireless device-to-device caching
  networks: {Basic} principles and system performance,'' \emph{{\emph{submitted
  to}} IEEE Journal on Sel. Areas in Commun.}, 2014, available online:
  arxiv.org/abs/1305.5216.

\bibitem{6787081}
N.~Golrezaei, P.~Mansourifard, A.~Molisch, and A.~Dimakis, ``Base-station
  assisted device-to-device communications for high-throughput wireless video
  networks,'' \emph{IEEE Trans. on Wireless Commun.}, vol.~13, no.~7, pp.
  3665--3676, Jul. 2014.

\bibitem{ji2014fundamental}
M.~Ji, G.~Caire, and A.~F. Molisch, ``Fundamental limits of caching in wireless
  {D2D} networks,'' \emph{\emph{submitted to} IEEE Trans. on Info. Theory},
  2014, available online: arxiv.org/abs/1405.5336.

\bibitem{Altieri}
A.~Altieri, P.~Piantanida, L.~Rey~Vega, and C.~Galarza, ``On fundamental
  trade-offs of device-to-device communications in large wireless networks,''
  \emph{IEEE Trans. on Wireless Commun.}, vol.~14, no.~9, pp. 4958--4971, Sep.
  2015.

\bibitem{ji2015throughput}
M.~Ji, G.~Caire, and A.~F. Molisch, ``The throughput-outage tradeoff of
  wireless one-hop caching networks,'' \emph{{\emph{submitted to}} IEEE Trans.
  on Commun.}, 2015, available online: arxiv.org/abs/1312.263.

\bibitem{gupta2000capacity}
P.~Gupta and P.~R. Kumar, ``The capacity of wireless networks,'' \emph{IEEE
  Trans. on Info. Theory}, vol.~46, no.~2, pp. 388--404, Mar. 2000.

\bibitem{haenggi2012stochastic}
M.~Haenggi, \emph{Stochastic Geometry for Wireless Networks}.\hskip 1em plus
  0.5em minus 0.4em\relax Cambridge University Press, 2012.

\bibitem{baccelli2009stochastic}
F.~Baccelli and B.~Blaszczyszyn, \emph{Stochastic Geometry and Wireless
  networks{,} Volume 1- Theory}.\hskip 1em plus 0.5em minus 0.4em\relax NOW:
  Foundations and Trends in Networking, 2009.

\bibitem{mukherjee2014analytical}
S.~Mukherjee, \emph{Analytical Modeling of Heterogeneous Cellular
  Networks}.\hskip 1em plus 0.5em minus 0.4em\relax Cambridge University Press,
  2014.

\bibitem{dhillon2012modeling}
H.~S. Dhillon, R.~K. Ganti, F.~Baccelli, and J.~G. Andrews, ``Modeling and
  analysis of {$K$}-tier downlink heterogeneous cellular networks,'' \emph{IEEE
  Journal on Sel. Areas in Commun.}, vol.~30, no.~3, pp. 550--560, Apr. 2012.

\bibitem{mukherjee2012distribution}
S.~Mukherjee, ``Distribution of downlink {SINR} in heterogeneous cellular
  networks,'' \emph{IEEE Journal on Sel. Areas in Commun.}, vol.~30, no.~3, pp.
  575--585, Apr. 2012.

\bibitem{novlan2013analytical}
T.~D. Novlan, H.~S. Dhillon, and J.~G. Andrews, ``Analytical modeling of uplink
  cellular networks,'' \emph{IEEE Trans. on Wireless Commun.}, vol.~12, no.~6,
  pp. 2669--2679, Jun. 2013.

\bibitem{elsawy2014stochastic}
H.~ElSawy and E.~Hossain, ``On stochastic geometry modeling of cellular uplink
  transmission with truncated channel inversion power control,'' \emph{IEEE
  Trans. on Commun.}, vol.~13, no.~8, pp. 4454--4469, Aug. 2014.

\bibitem{lin2013comprehensive}
X.~Lin, J.~G. Andrews, and A.~Ghosh, ``Spectrum sharing for device-to-device
  communication in cellular networks,'' \emph{IEEE Trans. on Wireless Commun.},
  vol.~13, no.~12, Dec. 2014.

\bibitem{elsawy2014analytical}
H.~ElSawy and E.~Hossain, ``Analytical modeling of mode selection and power
  control for underlay {D2D} communication in cellular networks,'' \emph{IEEE
  Trans. on Commun.}, vol.~62, no.~11, pp. 4147--4161, Nov. 2014.

\bibitem{zhu2015joint}
K.~Zhu and E.~Hossain, ``Joint mode selection and spectrum partitioning for
  device-to-device communication: A dynamic stackelberg game,'' \emph{IEEE
  Trans. on Wireless Commun.}, vol.~14, no.~3, pp. 1406--1420, Mar. 2015.

\bibitem{mozaffari2015unmanned}
M.~Mozaffari, W.~Saad, M.~Bennis, and M.~Debbah, ``Unmanned aerial vehicle with
  underlaid device-to-device communications: Performance and tradeoffs,'' 2015,
  available online: arxiv.org/abs/1509.01187.

\bibitem{ShDhiC2015a}
S.~Krishnan and H.~S. Dhillon, ``Distributed caching in device-to-device
  networks: A stochastic geometry perspective,'' in {\em Proc. Asilomar},
  Pacific Grove, CA, Nov. 2015.

\bibitem{feng2014tractable}
H.~Feng, H.~Wang, X.~Xu, and C.~Xing, ``A tractable model for device-to-device
  communication underlaying multi-cell cellular networks,'' in \emph{Proc.,
  IEEE Intl. Conf. on Commun. (ICC)}, Jun. 2014.

\bibitem{sun2014d2d}
H.~Sun, M.~Wildemeersch, M.~Sheng, and T.~Q. Quek, ``{D2D} enhanced
  heterogeneous cellular networks with dynamic {TDD},'' \emph{{\emph{to appear,
  }}IEEE Trans. on Wireless Commun.}, 2015, available online:
  arxiv.org/abs/1406.2752.

\bibitem{george2014analytical}
G.~George, R.~K. Mungara, and A.~Lozano, ``An analytical framework for
  device-to-device communication in cellular networks,'' 2014, available
  online: arxiv.org/abs/1407.2201.

\bibitem{sakr2014cognitive}
A.~H. Sakr and E.~Hossain, ``Cognitive and energy harvesting-based {D2D}
  communication in cellular networks: Stochastic geometry modeling and
  analysis,'' \emph{IEEE Trans. on Commun.}, vol.~63, no.~5, pp. 1867--1880,
  May. 2015.

\bibitem{mungara2014spatial}
R.~K. Mungara, X.~Zhang, A.~Lozano, and R.~W. Heath~Jr., ``On the spatial
  spectral efficiency of {ITLinQ},'' in \emph{Proc., IEEE Asilomar}, Nov. 2014.

\bibitem{lin2013modeling}
X.~Lin, R.~Ratasuk, A.~Ghosh, and J.~G. Andrews, ``Modeling, analysis and
  optimization of multicast device-to-device transmissions,'' \emph{IEEE Trans.
  on Wireless Commun.}, vol.~13, no.~8, pp. 4346--4359, Aug. 2014.

\bibitem{pyattaev2015understanding}
A.~Pyattaev, O.~Galinina, S.~Andreev, M.~Katz, and Y.~Koucheryavy,
  ``Understanding practical limitations of network coding for assisted
  proximate communication,'' \emph{IEEE Journal on Sel. Areas in Commun.},
  vol.~33, no.~2, pp. 156--170, Feb. 2015.

\bibitem{andreev2014analyzing}
S.~Andreev, O.~Galinina, A.~Pyattaev, K.~Johnsson, and Y.~Koucheryavy,
  ``Analyzing assisted offloading of cellular user sessions onto {D2D} links in
  unlicensed bands,'' \emph{IEEE Journal on Sel. Areas in Commun.}, vol.~33,
  no.~1, pp. 67--80, 2015.

\bibitem{altieri2014fundamental}
A.~Altieri, P.~Piantanida, L.~Vega, and C.~Galarza, ``On fundamental trade-offs
  of device-to-device communications in large wireless networks,''
  \emph{\emph{to appear, }IEEE Trans. on Wireless Commun.}, 2015.

\bibitem{zhang2014social}
Y.~Zhang, E.~Pan, L.~Song, W.~Saad, Z.~Dawy, and Z.~Han, ``Social network aware
  device-to-device communication in wireless networks,'' \emph{IEEE Trans. on
  Wireless Commun.}, vol.~14, no.~1, pp. 177--190, Jan. 2015.

\bibitem{hu2014evaluating}
X.~Hu, L.~Meng, and A.~D. Striegel, ``Evaluating the raw potential for
  device-to-device caching via co-location,'' \emph{Procedia Computer Science},
  vol.~34, pp. 376--383, 2014.

\bibitem{MehrnazD2D1}
M.~Afshang, H.~S. Dhillon, and P.~H.~J. Chong, ``Modeling and performance
  analysis of clustered device-to-device networks,'' \emph{\emph{submitted to}
  IEEE Trans. on Wireless Commun.}, available online: arxiv.org/abs/1508.02668.

\bibitem{AfsDhiC2015}
------, ``Coverage and area spectral efficiency of clustered device-to-device
  networks,'' in {\em Proc. IEEE Globecom}, San Diego, CA, Dec. 2015.

\bibitem{DalVerB2003}
D.~J. Daley and D.~Vere-Jones, \emph{An Introduction to the Theory of Point
  Processes. {Volume} I: Elementary Theory and Methods}, 2nd~ed.\hskip 1em plus
  0.5em minus 0.4em\relax New York: Springer-Verlag, 2003.

\bibitem{ganti2009interference}
R.~K. Ganti and M.~Haenggi, ``Interference and outage in clustered wireless ad
  hoc networks,'' \emph{IEEE Trans. on Info. Theory}, vol.~55, no.~9, pp.
  4067--4086, Sep. 2009.

\bibitem{david1970order}
H.~A. David and H.~N. Nagaraja, \emph{Order Statistics}.\hskip 1em plus 0.5em
  minus 0.4em\relax New York: John Wiley and Sons, 1970.

\bibitem{caflisch1998monte}
R.~E. Caflisch, ``Monte carlo and quasi-monte carlo methods,'' \emph{Acta
  Numerica}, vol.~7, pp. 1--49, Jan. 1998.

\bibitem{AfsDhiC2015a}
M.~Afshang and H.~S. Dhillon, ``Spatial modeling of device-to-device networks:
  {Poisson} cluster process meets {Poisson} hole process,'' in {\em Proc.
  Asilomar}, Pacific Grove, CA, Nov. 2015.

\bibitem{short2012computation}
R.~T. Short, ``Computation of rice and noncentral chi-squared probabilities,''
  Apr. 2012.

\end{thebibliography}

\end{document}